\documentclass[11pt,letterpaper]{article}
 \usepackage[margin=1in]{geometry}
 \usepackage{hyperref, authblk}
 \hypersetup{
     colorlinks=true,
     linkcolor=blue,
     citecolor=blue,
 }

\usepackage{bm,xfrac,amsmath,amsthm,amssymb, soul, comment, xspace}
 \usepackage[dvipsnames]{xcolor}
\newcommand{\addreferencesection}{
  \phantomsection
  \addcontentsline{toc}{section}{References}
  }

\newcommand{\netco}{{network-coordination game}\xspace}
\newcommand{\netcos}{{network-coordination games}\xspace}
\newcommand{\alert}[1]{} %{\textcolor{red}{\bf #1}}

\newcommand{\payoff}{{\mathrm{payoff}}}
\newcommand{\sss}{{{\bm \sigma}}}

\newcommand{\CP}{{{\mathscr P}}}
\newcommand{\CQ}{{{\mathscr Q}}}

\newcommand{\CA}{{{\mathscr A}}}
\newcommand{\Li}{{transformation vector}\xspace}
\newcommand{\Lset}{{transformation set}\xspace}
\newcommand{\Lcal}{{\mathscr{L}}}

\title{Smoothed Efficient Algorithms and Reductions for Network Coordination Games}
 \author{Shant Boodaghians}
 \author{Rucha Kulkarni}
 \author{Ruta Mehta\footnote{This work was supported by NSF grant CCF-1750436}}

 \affil{Department of Computer Science,\\ 
 University of Illinois at Urbana-Champaign,\\
 \texttt{\{boodagh2,ruchark2\}@illinois.edu, rutamehta@cs.illinois.edu} 
 }
\date{}

\newtheorem{thm}{Theorem}[section]
\newtheorem{lem}[thm]{Lemma}
\newtheorem{claim}[thm]{Claim}

\newtheorem{cor}[thm]{Corollary}

\newtheorem*{claim*}{Claim}

\theoremstyle{definition}

\newtheorem{defn}{Definition}[section]

\newcommand{\ip}[1]{\left\langle #1 \right\rangle}
\newcommand{\R}{\mathbb R}
\newcommand{\kset}{\{1,\,\dotsc,\,k\}}
\newcommand{\ts}{\textsuperscript}
\newcommand{\nashcoord}{{NetCoordNash}\xspace}
\newcommand{\maxcut}{{FlipMaxCut}\xspace}

\usepackage{tikz}
\usepackage[mathscr]{euscript}
 \usepackage{natbib}
\begin{document}

 \thispagestyle{empty}
\maketitle

\begin{abstract}

%Extensive work in the last two decades has brought deep insights into the worst case complexity of computing Nash equilibria (NE). However, largely negative results have raised the need for {\em beyond worst-case} analysis of these problems. 

Worst-case hardness results for most equilibrium computation problems have raised the need for {\em beyond-worst-case} analysis.
To this end, we study the smoothed complexity of finding pure Nash equilibria in Network Coordination Games, a PLS-complete problem in the worst case. This is a potential game where the {\em sequential-better-response} algorithm is known to converge to a pure NE, albeit in exponential time. First, we prove polynomial (resp. quasi-polynomial) smoothed complexity when the underlying game graph is a complete (resp. arbitrary) graph, and every player has constantly many strategies. We note that the complete graph case is reminiscent of perturbing {\em all parameters}, a common assumption in most known smoothed analysis results.

\vspace*{2pt}

Second, we define a notion of {\em smoothness-preserving reduction} among search problems, and obtain reductions from $2$-strategy network coordination games to local-max-cut, and from $k$-strategy games (with arbitrary $k$) to local-max-cut up to two flips. The former together with the recent result of \cite{BCC18} gives an alternate $O(n^8)$-time smoothed algorithm for the $2$-strategy case. %known smoothed efficient algorithm for the local-max-cut problem, gives an alternate efficient smoothed algorithm for the $2$-strategy games. 
This notion of reduction allows for the extension of smoothed efficient algorithms from one problem to another. 

%For the first set of results, we combine and extend the local-max-cut approaches \cite{ER14,A+17} to handle the multi-strategy case, where nodes are replaced with (player,strategy) pairs. Our approach is baѕed on a careful choice of matrix to capture potential change under better-response, lower bound its rank, and combin it with an appropriate union bound -- the last two requires a careful case analysis. We believe our approach should be of interest to address smoothed complexity of other potential/congestion games. 

% For the first set of results, we develop techniques to analyze the probability of a slow increase in potential during the better-response algorithm for a perturbed game. Our approach carefully choses a matrix to capture potential changes, lower bounds its rank under various cases, and combines it with an appropriate union bounds. This combines and extends the local-max-cut approaches \cite{ER14,A+17} to handle the multi-strategy case, and therefore we believe should be of interest to address smoothed complexity of other potential/congestion games. 
% \medskip

% OR
% \medskip

For the first set of results, we develop techniques to bound the probability that an (adversarial) better-response sequence makes slow improvements on the potential. Our approach combines and generalizes the local-max-cut approaches of \cite{ER14,A+17} to handle the multi-strategy case: it requires a careful definition of the matrix which captures the increase in potential, a tighter union bound on adversarial sequences, and balancing it with good enough rank bounds. We believe that the approach and notions developed herein could be of interest in addressing the smoothed complexity of other potential and/or congestion games. 
\end{abstract}

\newpage
 \thispagestyle{empty} 
 \setcounter{tocdepth}{2}
 \tableofcontents
 \newpage
 \pagenumbering{arabic}

\section{Introduction}\label{sec:introduction}
Nash equilibrium is one of the most central solution concepts of game theory. 
Extensive work within Algorithmic Game Theory has brought significant insights to the computational complexity of finding Nash equilibria (NE) in various game models (see Section \ref{sec:related-work} for a detailed discussion). %\cite{Papadimitriou94,LMM03,AKV05,DGP06,CDT06,CS03,DGP09,KT10,EY10,  Mehta14, MVY15, CDO15, FPT04, SV08}.
Most problems of this form are shown to be complete for some class in ``Total-Function NP''\footnote{TFNP, A class of search problems in the intersection of NP and co-NP.}, typically for either PPAD or PLS {\em e.g.,} \cite{FPT04,DGP09,CDT06Smooth,CDT06,SV08,KT10,EY10,S10,Mehta14}.
The class PPAD captures problems with parity arguments like finding fixed-points of functions, Sperner's Lemma, and finding (mixed) Nash equilibria in general games \cite{Papadimitriou94,Kintali13,DGP09,CDT06Smooth,Goldberg11}, while PLS (Polynomial Local Search) captures problems with local-search algorithms, like local-max-cut, local-max-SAT, and pure NE in potential games \cite{JPY88, SY91, FPT04, SV08, CD11}.

%In particular, finding NE in the simplest model of {\em two-player} games is PPAD-complete \cite{}; such a game is represented by two payoff matrices $(A,B)$. The problem remains hard for various special sub-classes, as well as to approximate \cite{}. A well-studied multi-player extension of this is of {\em network game}, defined on a graph where each node is a player playing a two-player game with each of its neighbors. These being more general than two-player games finding NE remains PPAD-complete, even constant approximation \cite{}. On the other hand, a {\em network coordination game}, where the two-player game on every edge is a coordination game, {\em i.e.,} game $(A,B)$ such that $A=B$, then it becomes a potential game and is guaranteed to have a pure NE\footnote{A NE where every player chooses a strategy to play deterministically}. However, finding one is PLS-complete \cite{}, and remains so for various other potential/congestion games \cite{}. Extensive work on both general and potentail games have brought a fair understanding of worst case complexity of both PPAD and PLS from various perspectives, {\em e.g.,} query complexity \cite{}, communication complexity \cite{}, and decision problems \cite{}.

%\alert{-- sentence about potential games --}
 Although it is well accepted that PPAD and PLS are unlikely to be in P \cite{BCEIP98,BPR15,Rub17ETH}, problems in these classes admit respectively path-following style complementary pivot algorithms \cite{LH64,GW03,AGMS11,GJM11} and local-search-type algorithms \cite{JPY88}.
 The natural local-search algorithms for PLS problems have been observed to be empirically fast \cite{JPY88,CDP08,DFIS16}.
 %\footnote{Johnson, Papadimitriou, and Yannakakis (1988) said, ``Practically all the empirical evidence would lead us to conclude that finding locally optimal solutions is much easier than solving NP-hard problems''}. 
 However, these algorithms take exponential time in the worst case \cite{SY91,SvS04}. 
A similar phenomenon occurs with the classical Simplex method for solving Linear Programs. To study this case, Spielman and Teng introduced a powerful model of {\em Smoothed analysis}, which ``continuously interpolates
between the worst-case and average-case analyses of algorithms,'' \cite{ST04}. 
The basic idea is to formally show that adversarial instances are ``sparse and scattered,'' in a probabilistic sense. %given any (adversarial) instance, once {\em small perturbations} are introduced to it, the algorithm will be efficient in expectation or with high probability. 
%
%These classes also capture a number other problems from diverse fields, such as optimization and cryptography, that are believed to be hard \cite{}, and therefore by now it is well acceted that they are unlikely to be in P \cite{BPR, Rub}.
This gives rise % for practitioners/game theorists who believe that most real-life systems do operate at near equilibrium \cite{}, giving rise 
to the following question:

\vspace*{2pt}
\paragraph{Question.} Can we design smoothed efficient algorithms for finding Nash equilibria? %that are efficient under some ``beyond worst-case'' model, e.g., smoothed analysis, average case, etc.?
\medskip

In this paper we answer the question in the affirmative for network-coordination games, a well-studied model (see $e.g.$ \cite{CD11} and \cite{SW16}) which succinctly captures pairwise coordination in multi-player games. % under the smoothedness model, where parameters of the given instance are perturbed randomly within a given density. 
We obtain smoothed (quasi-)polynomial time algorithms to find pure Nash equilibria (PNE) in network-coordination games (\nashcoord) with constantly many strategies, a PLS-complete problem in the worst case \cite{CD11}. 
%To the best of our knowledge this gives the first smoothed efficient algorithm for a ``hard'' Nash equilibrium problem.

\paragraph{Smoothed Analysis.} The work of Spielman and Teng \cite{ST04} introduced the smoothed analysis framework to study good empirical performance of classical Simplex method for linear programs (LP). 
%Smoothed analysis was introduced by Spielman and Teng to understand empirically fast running time of the classical simplex algorithm on linear programes (LP), which is exponential in the worst case \cite{}. 
%
%They reason that, if the algorithm performs badly only on sparsely distributed instances, then this would explain our observation that the Simplex method performs well in practice. 
They showed that introducing independent random perturbations to any given (adversarial) LP instance, ensures that the Simplex method terminates fast with high probability (here the run-time depends inverse polynomially in the perturbation magnitude). 
%
%when independently random perturbations are introduced into {\em all} entries of the constraint matrix (referred to as a {\em smoothed instance}), then the Simplex method, implemented with the shadow-pivot rule, finds a solution in polynomial time, both in expectation as well as with high probability. Here, the running time of the algorithm depends polynomially on the density of the perturbation: as the variance of the perturbation approaches zero, the density grows unboundedly, and the instance approaches a non-perturbed instance. %This essentially tells us that Simplex is efficient on almost all LP instances, except for sparsely distributed and scattered instances. 
%
Performance on such probabilistic instance has since been known as {\em smoothed complexity} of the problem --  
%which has been accepted as 
one of the strongest guarantees one can hope for {\em beyond worst-case}. 
In the past decade and a half, much work has sought to obtain smoothed efficient algorithms when worst-case efficiency seems infeasible \cite{DMRSS03, BV04, MR05, RV07,  AV09, ERV14, ER14, A+17}, including for integer programming, binary search trees, iterative-closest-point (ICP) algorithms, the $2$-OPT algorithm for the Traveling Salesman problem (TSP), the knapsack problem, and the local-max-cut problem. 

In case of Nash equilibrium (NE) computation, smoothed complexity of two-player games is known to not lie in P unless RP$\:=\:$PPAD \cite{CDT06Smooth}, which
follows from the hardness of 1/poly additive approximation. %and the fact that solution of a perturbed instance gives approximate solution of the original instance. 
On the contrary, for most PLS-complete problems, the natural local-search algorithm often finds an additive-approximate solution efficiently. There is always a ``potential function'' that the algorithm improves in each step. Intuitively, until an approximate solution is reached, the algorithm will improve the associated potential function significantly in every local-search step.

The potential function is also reminiscent of the {\em objective function} of LP's. %which must increase with every step of the Simplex method. 
This allows one to perform smoothed analysis of FLIP algorithm for {\em local-max-cut}, a classical PLS-complete problem. Here, given a weighted graph, the goal is to find a cut that can not be improved by moving a single vertex across the cut. The algorithm FLIPs the partition of any one vertex in one step if this improves the cut.\footnote{Note that, unlike max-cut, if all edge weights are $poly(n)$, assuming they are integers, then FLIP finds a local-max-cut in polynomial time. This indicates the existence of many ``easy'' instances ``near'' any given instance. Similar is the case with other PLS-complete game problems too.} 
First, \cite{ER14} showed that FLIP terminates in quasi-polynomial time with high probability, when the edge weights are perturbed. 
A recent second result by Angel, Bubeck, Peres, and Wei \cite{A+17} showed polynomial run-time for the same algorithm when the weights of {\em all} edges are perturbed, viewing missing edges as zero-weight edges. 
We note that the simultaneous perturbation of {\em all} input parameters seems to be crucial in getting smoothed polynomial time algorithms so far, {\em e.g.,} see \cite{ BD02, ST04,SST06}.

\paragraph{Summary of Our Results.} Motivated by the above intuition and results, we study the smoothed complexity of \nashcoord, which we recall is also PLS-complete. An instance of \nashcoord is represented by an undirected {\em game graph} $G=(V,E)$, where the nodes are the players, and every node $v\in V$ plays a two-player coordination game with each of its neighbors. If every player has $k$ strategies to choose from, then the game on each edge $(u,v)$ can be represented by a $k\times k$ payoff matrix $A_{uv}$. Once every player chooses a strategy, the payoff value for each edge is fixed, and each player gets the sum of the payoffs on its incident edges. The goal is to find a pure NE of this game. We analyze the problem through two different approaches: $(i)$ direct smoothed analysis of an algorithm (see Section~\ref{sec:BRA}), and $(ii)$ through reductions (Section~\ref{sec:red}). 
%A strategy profile is said to be at a Nash equilibrium if no player gains by deviating unilaterally. 
%

%It is easy to see that 
We first analyze a local-search algorithm called {\em better-response} (BRA), where players take turns making improving moves (termed {\em better-response} moves).
%, while all other players remain fixed. 
The sum of payoffs across all the edges acts as a {\em potential function} measuring the progress of the algorithm, $i.e.$ the function value increases every time a player makes a better-response move %. while all other players remain fixed 
(Section~\ref{sec:PrelGame}). % for details).
%
%The sum of payoffs across all the edges acts as a {\em potential function}, $i.e.$ the function value increases every time a player makes an improving move (termed ``{\em better-response} move''), while all other players remain fixed (see Section \ref{sec:PrelGame} for details). 
%Therefore, any algorithm where the players take turns making improving moves must output a pure NE (PNE) of the game. 
%We call such an algorithm a {\em better-response} algorithm (BRA). 

We show that, when $k$ is constant, then for games where for all $(u,v)\in E$, all payoff entries of $A_{uv}$ are perturbed independently at random, %within density $\phi$ 
any {\em better-response} algorithm converges in quasi-polynomial time with high probability (Theorem~\ref{thm:in-text-qpoly}). 
Furthermore, if $G$ is a complete graph and any two players participate in a random game $A_{uv}$, 
% or we consider it as a complete graph by taking $A_{uv}$ to be an all zero matrix for the edges that are not present, 
then the algorithm takes polynomial time with high probability (Theorem~\ref{thm:whp-poly}). 
To show these results, we develop a technique to bound the probability that a sufficiently long sequence of better-response moves makes little improvement to 
%the value of 
the potential function. 
This technique may be applicable to analyze the smoothed complexity of other potential games. 
%Towards this, we show that our approach gives smoothed quasi-polynomial algorithm for a congestion games when \# strategies/player and \#players/resource both are constants. 
%To the best of our knowledge, ours is the first smoothed efficient algorithm for a Nash equilibrium problem that is hard in the worst-case.
Apart from party affiliation games, for which efficient smoothed complexity directly follows from local-max-cut \cite{FPT04}, % for party affiliation games that directly follows The efficient smoothed complexity of local-max-cut directly implies the same for NE computation in party affiliation games \cite{FPT04}. %These games have a restrictive model where every player has $2$ strategies to choose from, and the problem can be easily reduced to the local-max-cut problem (the players map to the vertices, the edges to the games, and the sides of the cut to the strategies chosen. The sum of payoffs of all players (up to a factor $2$) is the cut value in the reduced local-max-cut problem). 
%Apart from this, 
to the best of our knowledge, no smoothed efficient algorithm for a worst-case hard Nash equilibrium problem was known prior to this work. 
%that was also hard in the worst case.
%We denote as $k$-\nashcoord instances of \nashcoord for fixed $k$.

\begin{figure}\label{fig:maxcut-to-netco}
\begin{center}
\begin{tikzpicture}
\fill (0,0) circle (2pt);
\fill (2,0) circle (2pt);
\draw (0,0)--(2,0);
\node at (0,0) [above left] {$u$};
\node at (2,0) [above right] {$v$};
\node at (1,0.3) [above] {$w_{uv}$};

\draw[thick] (3,0.5)++(0,3pt)--+(0,-6pt);
\draw[thick] (3,0.5)--(4,0.5)--+(-4pt,3pt)--+(0,0)--+(-4pt,-3pt);

\fill (5,0) circle (2pt);
\fill (9,0) circle (2pt);
\draw (5,0)--(9,0);
\node at (5,0) [above left] {$u$};
\node at (9,0) [above right] {$v$};
\node at (7,0.2) [above] {$A_{uv}=\left[\begin{smallmatrix}
	0&w_{uv}\\[2pt]w_{uv}&0
\end{smallmatrix}\right]$};
\end{tikzpicture}
\end{center}
\vspace{-0.5cm}
\caption{Local-max-cut to $2$-strategy \netcos: mapping of edge $(u,v)$.}
\vspace{-0.5cm}
\end{figure}

Our analysis combines and extends the approaches for local-max-cut \cite{A+17,ER14}; as discussed in Section~\ref{sec:challenge}. Local-max-cut reduces to a special case of $2$-strategy \netcos \cite{CD11} (see Figure~\ref{fig:maxcut-to-netco}), therefore handling the more general setting of \netco with multiple strategies poses some challenges.

Our second set of results analyzes smoothed complexity through reductions. To extend a smoothed efficient algorithm from one problem to another, the usual notion of reduction does not suffice. Among other things, it needs to ensure that independent perturbations of the parameters of the original problem produce independent perturbations of all parameters in the reduced problem. Based on this we define {\em smoothness-preserving} reduction (see Theorem~\ref{lem:reductionGivesAlg}), and obtain two such reductions:
%For us, the obvious choice of problem to reduce to is local-max-cut, however now the reduction has to be {\em from} \netco  { \em to} local-max-cut. Furthermore, the reduction should ensure that independent perturbations in the parameters of the former produce independent perturbation of all edge weights in the latter problem, among other things. We call such a reduction {\em smoothness-preserving}, and obtain two such reductions:
$(i)$~from 2-strategy \netcos to local-max-cut, and $(ii)$ general $k$-strategy \netcos to local-max-cut up-to-2-flips. ($i.e.$ cuts whose value cannot be improved by moving any {\em two} vertices) (see Theorem~\ref{thm:reduction}). 
We show that the first reduction, together with smoothed efficient algorithms for local-max-cut gives alternate smoothed efficient algorithms for the $2$-strategy \netco; in particular, the recent result of \cite{BCC18} gives an $O(n^8)$ algorithm.
For general \netcos, polynomial smoothed complexity of local-max-cut up-to-two-flips, where the local improvement algorithm flips two vertices in every step, needs to be shown. We leave this as an open question.

\subsection{%\color{red}
Smoothed Algorithms: Challenges and Techniques}
%Overview of challenges and new techniques}
\label{sec:challenge}
%\Rucha{have to shorten/polish}
%{\color{red}%As we note in section~\ref{sec:common}, 
Our results follow a framework which is common to past work on smoothed algorithms for local max cut~\cite{ER14,A+17,BCC18} (see Section~\ref{sec:common}). The goal is to show that with high probability {\em every} sufficiently long sequence of improving moves (of the local-search algorithm) is very likely to significantly improve the potential. This is shown via a two-step process %All use the following general framework 
hinging on Lemma~\ref{lem:probability}: first, we represent the potential improvement in every step as a linear combination of the input parameters and consider the corresponding matrix for a subsequence, and second, show a tight union bound on the number of different sequences and relevant initial configurations, paired with a high rank-bound for the matrix. %lower bound the rank of the matrix and combine it with a union bound over all possible improving subsequences and relevant initial configurations.

Several obstacles were encountered while trying to apply the general framework to the $k$-\nashcoord problem. The first was how to correctly represent each move in a better-response (BR) sequence. Specifying only the moving player is insufficient, as every player has more than $2$ strategies to choose from. 
A move is entirely specified by the triple (player, strategy-from, strategy-to). 
This, however, is too descriptive and the union bound is too large. 
Labelling as (player, to-strategy) suffices, and strikes the right balance between rank and union bound. This definition, however, muddies the proof technique further on.

\begin{itemize}
\item \textbf{Rank analysis:} When the BR sequence represents moves as (player, strategy) pairs, within a sequence, some {\em pairs} may be repeated, and some {\em players} may be repeated, but always playing distinct strategies. Therefore, notions of repeating and non-repeating players need to be carefully defined. 
Secondly, we would want the nodes in the directed-influence-graph arguments of Section~\ref{sec:cyclic}, used for showing rank bounds in terms of repeat-moves, to also be labeled as (player,strategy) pairs. This would have multiple nodes in the graph corresponding to the same players, each influencing multiple players. This makes for messy analysis and poor rank bounds. We instead label the nodes of these influence graphs as simply the players, unlike the improving sequence itself. 
This causes rank bounds to depend on the number of players with repeating moves, rather than the number of repeating moves. Thus, notions of {\em critical subsequences} and {\em separated blocks} need to be carefully adapted from past notions.

\item \textbf{Defining \textit{critical subsequences} and \textit{separated blocks}:} 
We show that moves by non-repeating players allow us to ``separate'' the sequence and combine the rank bounds from both sides. This leads to the notion of {\em separated blocks,} which requires a careful selection of the boundary moves. 
To combine these bounds, we use the idea of a {\em critical block} from~\cite{A+17}, which was very important to their poly-time bound, and is similarly helpful in our result.
It would have been preferable to use any non-repeated move as the separators, but it is not possible to do so, as the same player may also repeat other moves.
These notions must be adapted carefully so that the rank and union bounds balance each other. This, overall, loosens the rank bound.
%The high level ideas however, are clever and can be reformed to be applied in this case too, but lead to looser rank bounds.

%Although this led to higher ranks for the local-max-cut problem than $k$-\nashcoord, when properly designed, it allows converting the rank bound in terms of players, to one in terms of the length of the sequence. 
%\item Defining \textit{transition blocks}: To allow such a conversion from bounds in terms of players or moves to those in terms of the length of the sequence, one must also define \textit{transition block}-like subsequences. Defined as sequences of repeating players in \cite{A+17}, for $k$-\nashcoord we must remember the difference between repeating players and repeating moves in the sequence. For easier analysis,   we must have looser bounds, which hurts the minimum rank guarantee of a long enough BR sequence. This then leads to the next challenge.
\item \textbf{Union bounds:} The union bound analysis must now bear the brunt of the looser rank bounds above, and must be made tighter through properties of the $k$-\nashcoord problem. Eliminating the influence of the inactive players allows us to avoid having to take their strategies into account. Taking a sum of repeated moves does not suffice in our setting, and we define the notion of ``cyclic sums'' to handle this. 
\end{itemize}

%Overall, a number of new technical insights are needed to handle  multiple strategies per player. %in combining both and then extending to the network coordination game, which is clearly much more general than even local-max-k-cut. 
%Our approach may shed light on obtaining smoothed efficient algorithms for Nash equilibrium computation in general potential games. 
%Towards this, we show that part of our analysis applies successfully to a subclass of congestion games, where every player has at most $k$ strategies, and every resource occurs in a bounded number of strategies (a PLS-complete problem). 
%Designing the $k$-\nashcoord model in a specific manner, performing more sophisticated analyses when required, so that the frameworks of \cite{ER14} and \cite{A+17} could be remodeled to apply to a significantly general problem is the main contribution of this paper. 

\subsection{Related Work}\label{sec:related-work}

The works most related to ours are \cite{ER14} and \cite{A+17} analyzing smoothed complexity of local-max-cut; see Section \ref{sec:challenge} for a detailed comparison. Independently, Bibaksereshkeh, Carlson, and Chandrasekaran \cite{BCC18} improved the running-time for the local-max-cut algorithm, and obtained smoothed polynomial and quasi-polynomial algorithms for local-$3$-cut and local-$k$-cut with constant $k$ respectively. 
%In local-$k$-cut the goal is to find a $k$-partition of the vertices such that the total weight of edges crossing any partition can not be improved by moving any one vertex. 
Reduction of Figure \ref{fig:maxcut-to-netco} easily extends to reduce local-$k$-cut to $k$-strategy \netcos, implying that the latter significantly generalizes the former. However, the reduction is not {\em smoothness-preserving}, and hence our smoothed efficient algorithm is not directly applicable to solve the local-$k$-cut problem.
%A common theme among these results, ours, as well as of Arthur and Vassilvitskii \cite{AV09} on $k$-means, and Englert, R\"{o}glin, and V\"{o}cking \cite{ERV14} on traveling salesman is to track improvement in some function which acts as a measure of progress.

\paragraph{Beyond-worst case complexity of NE computation.} The smoothed-efficient algorithm for local-max-cut directly gives one for party affiliation games~\cite{FPT04}. For two-player games which are non-potential in general, Chen, Deng, and Teng \cite{CDT06Smooth} ruled out polynomial smoothed complexity unless RP=PPAD. While, towards average case analysis B\'ar\'any, Vempala, and Vetta \cite{BVV07} showed that a game picked uniformly at random has a NE with support size two for both the players whp. The average case complexity of a random potential game was shown to be polynomial in the number of players and strategies by Durand and Gaujal \cite{DG16}.

\paragraph{Worst-case analysis.} Potential games and equivalently congestion games have been studied at length (e.g., \cite{Rosenthal73, MS96, RT02, FPT04, CMN05}), capturing routing and traffic situations (e.g., \cite{Smith79, DN84, Rgarden07, HS10, TMMA13,nw1}), and resource allocation under strategic agents (e.g., \cite{JT04,FT12,cong1}). Unlike general games, existence of the potential function ensures that these games always have a pure NE \cite{Rosenthal73}. 
Finding pure NE is typically PLS-complete \cite{FPT04,CD11}, while finding any NE, mixed or pure, is known to be in CLS \cite{DP11}, a class in the intersection of PPAD and PLS.
A remarkable collection of work studies the loss in welfare at NE through the notions of Price-of-Anarchy and Price-of-Stability (e.g., \cite{KP99,RT02,CK05,ADGMS06,ADKTWR08,spoa,poa1}). %To the best of our knowledge no smoothed complexity results are known for these games, and 
Our approach should help provide ways to obtain smoothed efficient algorithms for these games.

Worst case complexity of NE computation in general non-potential games has been studied extensively.
The computation is typically PPAD-complete, even for various special cases (e.g., \cite{AKV05,CDT06,Mehta14,inbal}) and approximation (e.g., \cite{CDT06Smooth,Rub16}). On the other hand efficient algorithms have been designed for interesting sub-classes (e.g., \cite{KT07,TS,Imm11,AGMS11,CD11,CCDP15,ADHLMS16,BB17,B18}), exploiting the structure of NE for the class to either enumerate, or through other methods such as parameterized LP and binary search. 
For two-player games, Lipton, Mehta, and Markakis gave a quasi-polynomial time algorithm to find a constant approximate Nash equilibrium \cite{LMM03}. Recently, Rubinstein \cite{Rub17ETH} showed this to be the best possible assuming exponential time hypothesis for PPAD, and Kothari and Mehta \cite{KM18} showed a matching unconditional hardness under the powerful algorithmic framework of Sum-of-Squares with oblivious rounding and enumeration. These results are complemented by communication \cite{RubCom1,RubCom2} and query complexity lower bounds \cite{Que2,Que1,Que3}.  
Lower bounds in approximation under well-accepted assumptions have been studied for the decision versions \cite{GZ,CS,HK,BKW,DFS}. %, like NP-hardness, SETH, and planted-clique, have also been studied for the decision versions of the problem, {\em e.g.,} if there exist more than one equilibrium, equilibrium with payoff at least $h$, etc. \cite{GZ,CS,HK,BKW,DFS}. 

\section{Game Model, Smoothed Analysis, and Statement of Results}
We introduce here the model we consider for Network Coordination Games, the notion of smoothed analysis, and state our main contributions.

\noindent\paragraph{Notation:} In what follows, $[k]$ denotes $\kset$, and $\ip{.,.}$ is the inner product. 

\subsection{Nash Equilibria in Network Coordination Games}\label{sec:PrelGame}
%-- network game, network coordination game, and Nash equilibrium
A game with two players, where each player has $k$ strategies, can be defined by two $k \times k$ payoff matrices $(A,B)$, one for each player. 
%When the first player plays $i\in [k]$ and the second plays $j\in[k]$, their respective payoffs are $A(i,j)$ and $B(i,j)$. 
It is called a {\em coordination game} if $A=B$. 
We assume without loss of generality that every player has the same number of strategies.
% Considering games where all players have the same number of strategies is without loss of generality (wlog) since we can add dummy strategies with bad payoffs.

\paragraph{$k$-\netco.}
A \netco is a multi-player extension of coordination games. 
The game is specified by an underlying undirected graph $G=(V,E)$, 
where the nodes are players, and each edge represents a two-player coordination game between its endpoints.
A $k$-\netco is where each player has $k$ strategies, and the edge games are represented by $k\times k$ matrices $A_{uv}$. 
%Each player is required to play the same strategy with all of its neighbours, and her payoff is the sum over all payoffs received in these games.
If $u$ plays $i\in[k]$ and $v$ plays $j\in [k]$ then both get payoff $A_{u,v}(i,j)$ on this edge; we will sometimes denote this by $A((u,i)(v,j))$ to disambiguate. %to specify which player is playing which strategy.
Nash equilibria are invariant to shifting and scaling of the payoffs, so $w.l.o.g.$ we assume every entry of $A_{uv}$ is contained in $[-1,1]$.  
% -- for simplicity of notations now on we will denote it by $A(u_i,v_j)$. 
%\note{All payoffs can be scaled by one value without affecting the solution. This is because doing so doesn't change the relative comparison between payoff values of a player, and a strategy that is better than others still remains so. Hence w}ithout loss of generality, all payoff values in $A$ are assumed to lie in the range $[-1,1]$.%, a standard assumption \note{($e.g.$ \cite{?}).}
Let $n=|V|$. A {\em strategy profile} is a vector $\sss\in[k]^n$ where each player chooses a strategy from $[k]$. %We denote $\sigma_u$ as the strategy assigned to player $u$. 
%For any strategy profile $\sss$, 
The payoff of player $u$ is then:
\[
\payoff_u(\sss) := \textstyle\sum_{v:\: uv\in E} A_{uv}(\sigma_u,\sigma_v)
\]
\paragraph{Nash Equilibrium.} At a  Nash equilibrium (NE) no player gains by deviating unilaterally. %{\em i.e.} by changing her strategy when strategies of all other players remain fixed. 
In general, NE strategy profiles may be randomized. However, a NE where every player chooses a strategy deterministically is called {\em pure Nash equilibrium} (PNE). Formally, strategy profile $\sss$ is a PNE, if and only if
%\[
$
\forall u \in V,\ \ \payoff_u(\sss) \ge \payoff_u(\sigma'_u,\sss_{-u}),\ \  \forall \sigma'_u \in [k],
$
%\]
where $\sss_{-u}$ denotes the strategies of all the players in $\sss$ except $u$.
In {\em Potential Games} \cite{Rosenthal73}, PNE's are known to always exist. 
Such a game admits a {\em potential function} which encodes the  individual ``progress'' of the players, {\em i.e.} $\exists g:[k]^n \rightarrow \mathbb R$ such that for all $\sss\in[k]^n$,
$
g(\sss) - g(\sigma'_u,\sss_{-u}) = \payoff_u(\sss) - \payoff_u(\sigma'_u,\sss_{-u}),\ \ \ \forall u, \forall \sigma'_u \in [k].%,\ \ \ \forall \sss \in [k]^n
$
Clearly, every local-maximum of $g$, {\em i.e.} $\sss$ such that $g(\sss) \ge g(\sigma'_u,\sss_{-u}),\ \forall u, \forall \sigma'_u\in[k]$, is a pure NE. % for the potential game. 
%The following is a well-known fact:

\begin{lem}[\cite{CD11}]
Network Coord.{ }Games are potential games with potential function % and the corresponding potential function is the total payoff of all the players (up to a factor $2$)
\begin{equation}\label{eq:pot}
\payoff(\sss) =\textstyle \sum_{(u,v) \in E} A_{uv}(\sigma_u,\sigma_v) =
\frac12\sum_{u\in V} \payoff_u(\sss)
\end{equation}
\end{lem}

Our goal is to find a pure NE for a given $k$-\netco.

\paragraph{Better-Response Algorithm (BR alg., or BRA).}
Another immediate consequence of being a potential game is that, for any strategy profile $\sss$, if some player $u$ can deviate to $\sigma'_u$ and improve her payoff, 
then the move $\sigma_u\mapsto\sigma'_u$ is termed a {\em better-response (BR) move} for player $u$ from strategy profile $\sss$. 
Clearly, under such a BR move, the potential function value increases, {\em i.e.,} 
% \[\payoff(s'_u,\sss_{-u})-\payoff(\sss) = \payoff_u(s'_u,\sss_{-u}) -\payoff_u(\sss)>0.\] 
$\payoff(\sigma'_u,\bm \sigma_{-u})-\payoff(\bm \sigma) = \payoff_u(\sigma'_u,\bm \sigma_{-u}) -\payoff_u(\bm \sigma)>0.$ 

Note that $\payoff(\sss)$ may only take $k^n$ possible values.
Hence, if a BR move is made whenever possible, the players must  converge to a local optimum of the potential function, or equivalently, to a pure NE of the game. 
This gives a local-search based {\em better-response} algorithm to solve $k$-\nashcoord. %$(a)$~Start at arbitrary strategy profile. $(b)$~If there is a player who can deviate and improve, change her strategy to an improving one. $(c)$~Repeat step $(b)$ until no player can make an improving move. 

% For the analysis purposes, we will view a strategy profile as a function $\sigma:V\rightarrow [k]$ and payoff values defined by $A$ as a $|E|k^2$ dimensional vector where entry $(uv,i,j)$ is $A(u_i,v_j)$. Then the potential value $\payoff(\sigma)$ can be thought of as a dot production $\langle B_\sigma, A\rangle$, where $B_\sigma$ is a $|E|k^2$ vector with $(uv, i, j)$ entry set to $1$ if $\sigma(u)=i$ and $\sigma(v)=j$, and $0$ otherwise. 

\subsection{Smoothed Analysis and Reductions}\label{sec:PrelSmooth}
The notion of smoothed analysis was %is a notion of beyond-worst-case analysis for problems with worst-case exponential-time algorithms that perform surprisingly well empirically. The notion 
introduced by Spielman and Teng %to ``continuously interpolatesbetween the worst-case and average-case analyses of algorithms'' 
\cite{ST04} to bridge the gap between average- and worst-case analysis. %To perform smoothed analysis of an algorithm, the input parameters are randomly sampled from some input distribution, and the running time of the algorithm is analyzed either in expectation or with high probability. \note{redundant with smoothed analysis explanation in intro. Can put the quote there and delete this para?}
For a search problem $\CP$ let $(I,X)$ be an instance where $I$ is possibly discrete information, and $X$ is a real-valued vector whose dimension depends on $I$. For example, in the case of a $k$-\netco, $I$ consists of the game graph $G$ and the number of strategies $k$, and $X$ is the payoff vector $A$. 

\begin{defn}[Polynomial Smoothed Complexity $w.h.p.$ or in expectation]\label{def:Smooth}
    Let $(I,X)$ be a random instance of $\CP$, where $I$ is chosen arbitrarily, and $X$ is a random real-valued vector whose entries are independent and have density at most $\phi$. If there exists an algorithm which solves arbitrary instances of $\CP$ in finite time, and for all $I$, solves the random input $(I,X)$ in time at most $(\phi\cdot |I|\cdot |X|)^c$ for some $c>0$ with probability at least $1-1/poly(\phi,|I|,|X|)$, where $|X|$ denotes the number of entries in $X$, then $\CP$ is said to have {\em polynomial smoothed complexity $w.h.p.$}
If the same holds in expectation, then $\CP$ is said to have {\em polynomial smoothed complexity in expectation.}
\end{defn}
Standard (Turing) reductions between two search problems $\CP$ and $\CQ$ are well-defined, and used to extend an algorithm of $\CQ$ to solve instances of $\CP$, or to imply hardness for $\CQ$ given hardness for $\CP$. 
We extend this notion to define {\em smoothness preserving reductions}. 
\begin{defn}[Strong and Weak Smoothness-Preserving Reductions]\label{def:SmoothRed}
A {\em randomized, smoothness-preserving reduction} from a search problem $\CP$ to $\CQ$ is defined by poly-time computable functions $f_1$, $f_2$, and $f_3$, and a real probability space $\Omega \subseteq \mathbb R^d$, such that,
\begin{itemize}
\item For any $(I,X)\in \CP$, 
    and for arbitrary $R\in \Omega$,
    $\big(f_1(I),\, f_2(X,R)\big)$ is an instance of 
    $\CQ$, such that all (locally optimal) solutions $\sss$ map to a solution $f_3(\sss)$ of $(I,X)$. %\note{Should we set $f_3$ to be linear? Because FIXP?}
    
\item Whenever the entries of $X$ and $R$ are drawn independently at random from distributions with density at most $\phi$, then $f_2(X,R)$ has entries which are independent random variables with density at most $poly(\phi,|X|,|R|)$. This is called a {\em strong} reduction.

\item If the entries of $f_2(X,R)$ instead of being independent are linearly independent combinations of entries of $X$ and $R$, and the density is similarly bounded, then call it a {\em weak} reduction.
\end{itemize}
\end{defn}

\paragraph{$d$-\maxcut.}\label{sec:PrelLMC}
%Part of our results include a reduction to the local-max-cut problem \cite{}. 
Extending the notion of local-max-cut \cite{SY91}, we define the $d$-\maxcut problem, given by an undirected graph $G=(V,E)$ with edge weights $w_{uv}$ for all $uv\in E$. 
The goal is to find a non-empty subset $S\subsetneq V$ of vertices, such that $\delta(S) := \sum_{uv\in E:u\in S,v\notin S} w_{uv}$ is a local optimum up to $d$-flips, {\em i.e.,} $\delta(S)\geq \delta(S')$ for all $S'\subset V$ such that $S$ and $S'$ differ by at most $d$ vertices, or $|(S\setminus S') \cup (S'\setminus S)|\le d$. %\triangle A)$ for all $A\subseteq V$ such that $1\leq|A|\leq d$, and where $\triangle$ denotes symmetric difference.\footnote{The set $A\triangle B$ is defined as the set of elements lying in exactly one of $A$ and $B$. So $A\triangle B = (A\setminus B)\cup(B\setminus A)$.}
Note that, $1$-\maxcut is the usual local-max-cut problem.

% the cut value across $(S,\bar{S})$, namely $\sum_{(u,v)\in E: u\in S, v\in \bar{S}} w_{uv}$, can not be improved by moving any vertex from $S$ to $\bar{S}$ or from $\bar{S}$ to $S$. 

% Similarly, we can also define {\em local-max-cut up to $d$-flips}, where the goal is to find a partition $(S,\bar{S})$, such that the cut value can not be improved by switching at most $d$ vertices. Here, think of $d$ as a constant.

For a constant $d$, the $d$-\maxcut problem admits a natural {\em FLIP algorithm}, which like the BR algorithm, will check whether there exists a local improvement, and move the candidate solution to the improved solution until no local improvement is possible.

\section{Overview of Our Results and Techniques}
\label{sec:PrelOR}

In this section, we give a high-level overview of the proof method and formally state the results of the paper. 
The following Sections~\ref{sec:BRA}--\ref{sec:red} provide the details for these proofs.
As discussed above, our results are twofold: first, we extend pre-existing proof methods to directly show that the BRA terminates in smoothed polynomial time, and second, we introduce the notion of a smoothness-preserving reduction which allows us to give an alternate algorithm for 2-\nashcoord, and a conditional one for $k$-\nashcoord.
We begin with a formal definition of the smoothed problem:

\vspace*{-3pt}
\paragraph{Smoothed $k$-\nashcoord.}
Given a $k$-\netco given as an undirected graph $G=(V,E)$ and an $(|E|k^2)$-dimensional payoff vector $A$, where the entries of $A$ are independent random variables supported on $[-1,1]$ with density at most $\phi$, find a PNE.%pure NE of the game. 
\smallskip

Since the real-valued input to the smoothed 
%$k$-\nashcoord 
problem is stochastic, the running time of an algorithm for it will be as well, and the guarantees will be stated either $w.h.p.$ or in expectation (Definition \ref{def:Smooth}). 
The proof of smoothed-poly running time falls within a framework which has been used in the past to show smoothed-polynomial running time for the FLIP algorithm in local-max-cut~\cite{ER14,A+17,BCC18}, and can also be applied to other local-improvement potential-descent algorithms. This section begins with an overview of this common framework, then proceeds to explain how the required bounds may be shown in our setting. 

\subsection{A Common Framework for Local-Improvement Algorithms}\label{sec:common}

%To highlight the framework common to our problem and local-max-cut, we 
Observe that the potential function for $k$-\nashcoord, as given in~\eqref{eq:pot}, is an integer linear combination of the payoff values, and is actually a 0--1 combination.
%In the case of our problem and local-max-cut, it is a 0--1 combination.
The following framework can be applied whenever this holds, as long as the potential function's range is also polynomially bounded.
It hinges on the following lemma:

\begin{lem}[\cite{Rog08}]\label{lem:main-praobability} Let $X\in \R^d$ be a vector of $d$ independent random variables where each $X_i$ has density bounded by $\phi$. Let $\alpha_1,\,\dotsc,\,\alpha_r$ be $r$ linearly independent vectors in $\mathbb{Z}^d$.
then the joint density of $(\ip{\alpha_i,X})_{i\in [r]}$ is bounded by $\phi^r$, and for any given $b_1,b_2,\dotsc \in \mathbb{R}$ and $\epsilon>0$,\\[-.5em] %In particular, if sets $J_i\subset \mathbb{R}$ have measure $\epsilon$, then
\begin{equation}\label{eq:prob-bound-intro}
\Pr\Big[\textstyle\bigwedge_{i=1}^r\ip{\alpha_i,X}\in [b_i,b_i+\epsilon]\Big] \leq (\phi\epsilon)^r
\end{equation}
\end{lem}

Here, $X$ constitutes the random inputs to the smoothed problem, and the $\alpha$ vectors represent the change in the potential function. 
Formally, if a potential function $\Phi(\bm  \sigma)$ is given for some problem, and $\bm \sigma_1,\bm \sigma_2,\bm \sigma_3,\dotsc$ denotes the walk through the state space given by some local-improvement algorithm,
then we have vectors $\alpha_i$ such that $\ip{\alpha_i,X}=\Phi(\bm \sigma_i)-\Phi(\bm \sigma_{i-1})$ for all $i$.
Setting $b_i=0$, the above lemma upper bounds the probability that {\em every} step of the local-improvement algorithm is only a small improvement \mbox{($\ip{\alpha,X}<\epsilon$)}, while still being an improvement \mbox{($\ip{\alpha,X}>0$)}.

Let $\mathcal E$ be the event in the probability statement, that is, $\textstyle\bigwedge_{i=1}^r\ip{\alpha_i,X}\in [0,0+\epsilon]$.
If $\mathcal E$ does not hold, and the sequence is indeed an improving one, then at least one of the improvements must be at least $\epsilon$. If $\mathcal E$ does not hold for {\em any} sequence of $\Omega(n)$ moves, then we can bound the running time of the iterative algorithm by
\[
	\frac n \epsilon\cdot \left(\max_{\sigma} \Phi(\bm \sigma) - \min_{\sigma}\Phi(\bm \sigma)\right)
\]
Finally, if $\Phi$ is bounded in a (quasi)polynomially-sized range, and $\epsilon$ is taken to be $1/(quasi)poly(n)$, then we conclude that the procedure runs in (quasi)polynomial time.

\paragraph{Rank Bound vs Union Bound.} Note that in order to get (quasi)polynomial running time with high probability, we must first upper-bound the probability of event $\mathcal E$, simultaneously for
{\em all} sequences of $\Omega(n)$ local-improvement moves, for which we simply take the union bound.
To counteract this large union bound, we must lower-bound the rank of the matrix $[\alpha_i]_{i=1}^n$.
It remains then to choose $\epsilon$ correctly to counteract the union bound, and proving the best rank in general.

This highlights the main technical challenge when applying the common framework: labelling the moves. 
If a move's label is too descriptive --- {\em e.g.} a full state vector --- then the union bound will be much too large. However, if a move's label is too vague --- {\em e.g.} denoting a move only by the player who is moving, not the strategy --- then the rank bound will not be large enough.  
We introduce, in the next section, those parameters which work in our setting.
At a high-level, our analysis follows the framework of the previous local-max-cut papers~\cite{ER14,A+17}. 
However, since players have multiple strategies, this poses some challenges for the technical details. 
We keep notation consistent whenever possible, to allow making analogies to past approaches.

\subsection{Notation}\label{sec:notation}
Recall the problem of $k$-\nashcoord defined above, and the better-response algorithm (BRA) discussed in Section~\ref{sec:PrelGame}.
We represent each better-response (BR) move by a player-strategy pair $(u,i)$, denoting that player $u$ is replacing strategy $\sigma_u$ by $i$, (assuming $i\neq \sigma_u$). 
We also denote as~$\sss^t$ the strategy profile after the $t$\ts{th} BR move. 
Formally, $\sss^0$ is the initial strategy profile, 
and if the move at time $t$ is given by $(u_t,i_t)$, then
\(
	\sss^{t} := (i_t,\sss^{t-1}_{-u_t}).
\)
The change in the potential function at this step is then given by $\payoff(\sss^t) - \payoff(\sss^{t-1})$, which is clearly an integer linear combination of the $A_{uv}(i,j)$ payoff values. 
Since the combinations have integer coefficients, and the payoff values have density bounded by $\phi$, then the total improvement (a random variable) has density at most $\phi$ as well.
For any fixed BR sequence $S$ of length $2nk$, we define these linear combinations as the set of vectors 
$\Lcal=\{L_1,\,L_2,\,\dotsc\}$, where $L_t\in \{-1,0,1\}^{(|E|\times k^2)},\forall t\in [2nk]$, with entries indexed by each of the payoff values. 
The values of its entries are chosen as follows:
\begin{equation*}
L_t((u,i),(v,j))=\left\{\begin{array}{llll}
1 & \text{if: ~~}u_t\in \{u,v\}&\text{and ~~}\sigma^t_u=i&\text{and ~~}\sigma^t_v=j.\\
-1 & \text{if: ~~}u_t\in \{u,v\}&\text{and ~~}\sigma^{t-1}_u=i&\text{and ~~}\sigma^{t-1}_v=j.\\
0 & \text{otherwise.} 
\end{array} \right.
\end{equation*}
That is, every entry signifies whether the corresponding payoff value remains unchanged ($0$), or gets added ($1$) or removed ($-1$) from the potential function. 
The inner product $\ip{L_t,A}$ gives the change in the total payoff of the player who makes a move at time $t$, and thus, the increase in potential due to the $t$\ts{th} move.  

Each of the inner products $\ip{L_t,A}$ can be shown to be unlikely to take values in the range $(0,\epsilon]$ by the assumption of bounded density. To argue that $\Lcal$ has high rank, we partition all players who make a move in the sequence into two sets: 
those players who never play the same strategy twice throughout the whole sequence (non-repeating players), and those who do (repeating players). 
We will denote these quantities as $p_1$ and $p_2$ respectively, and define $p=p_1+p_2$.
Furthermore, since we will sometimes have to refer to players, and other times to moves, 
we denote as $d$ the number of distinct (player,strategy) pairs which appear in the sequence, and let $q_0$ denote the number of players which play a ``return move,'' that is, moves where a player returns to their original strategy.
Note that for any sequence of moves $S$, we have \[
p(S)\leq d(S)\leq k\cdot p(S),\qquad q_0(S)\leq p_2(S),
\qquad q_0(S)\leq d(S)/2. \]
We also introduce the quantity $d_1$, which is the number of (distinct) moves by all non-repeating players, so $p_1\leq d_1\leq k\cdot p_1$.

\subsection{Smoothed Polynomial Complexity, Rank Bounds, and Union Bounds}
\label{sec:smoothed-poly-overview}
Recall the definition of the smoothed $k$-\nashcoord defined above. 
The first main contribution of this paper is to show the following result:

\begin{thm}\label{thm:poly}
Given a smoothed instance of {\em $k$-\nashcoord} on a complete game graph, and with an arbitrary initial strategy profile, 
then any execution of BRA, where improvements are chosen arbitrarily, will converge to a PNE in at most 
$(nk\phi)^{O(k)}$ steps, with probability \mbox{$1-1/poly(n,k,\phi)$}.
\end{thm}
\vspace*{-1ex}

We also have convergence in $(nk\phi)^O({k})$ moves in expectation, as shown in theorem~\ref{thm:expectation}, discussed below.
Note that the probability value in the above statement is over the possible choices of payoff values for the network-coordination game, and not over executions of the BR algorithm. This statement holds true regardless of how BRA is implemented, even adversarially.
The {\em complete game graph} condition requires that any two players in this game share an edge in the game graph, and the payoff matrix for a non-existing game edge is not a fixed, all-zeros matrix, but is instead a random payoff matrix like all other edges.
This completeness technicality fits the model of most known smoothed polynomial-time algorithms ($e.g.$ \cite{ST04,SST06,ERV14,A+17,BCC18}) which require {\em every} parameter to be perturbed. 
We will later discuss results in the case of incomplete graphs, where missing edges are assumed to be 0-payoff games and are unperturbed.

Theorem~\ref{thm:poly} is shown using the ``common framework'' from Section~\ref{sec:common}. 
The random input is the set of (random) payoff matrices $\{A_{uv}\}_{u,v\in G}$, and the ``$\alpha$'' vectors are the columns of the $\mathcal L$ matrix as defined in Section~\ref{sec:notation} for a BR sequence of length $2nk$.
If all of the $n$ players appear in the sequence $S$,  then as shown in Corollary \ref{cor:all-active}, $\mathcal L$ has rank at least $\left(1-\tfrac 1n\right)\big(d(S)-q_0(S)\big) \geq n-1$.
Therefore, any sequence where every player is present is ``good'' with probability $(\phi\epsilon)^{n-1}$, where a sequence is ``good'' if it contains either a non-improving move, or a move which improves the potential by at least~$\epsilon$.

\paragraph{Case I. All Players Active.} Recall from Section~\ref{sec:common} that we wish to take the union bound over all sequences of length $\Omega(n)$. In fact, for this result, we will consider sequences of length $2nk$. Since there are $n$ players, there are $k^n$ possible initial configurations of the players, and $(nk)^{2nk}$ possible sequences, so with probability $1-(nk)^{O(nk)}(\phi\epsilon)^{\Omega(n)}$, all linear-length sequences are ``good.''
Setting $\epsilon= 1/\phi(nk)^{O(k)}$ suffices to have the probability of success be $1-1/poly(n,k)$.
Since $-n^2\leq \payoff(\sss)\leq n^2$, then with probability $1-1/poly(n)$, the BRA must terminate in at most $2n^2/\epsilon = \phi\cdot (nk)^{O(k)}$ many iterations, as desired.

\paragraph{Case II. Few Players Active.} However, it is not always the case that there are $n$ {\em active} players, or even $\Omega(n)$ active players in a given sequence of length $2nk$. A player is {\em active} if they appear in the sequence, and otherwise, {\em inactive}. 
We will show, in the following sections, rank bounds which depend on the $p_1$, $p_2$, $q_0$, and $d$ values, as defined in Section~\ref{sec:notation}. 
As we will discuss, these ranks will not be sufficiently large to handle the na\"ive union bound described above. The following table summarizes the bound we show in each case, and the resulting success probability, 
under the assumption $p(S)\leq \ell = 2(d(S)-q(S))\leq k\cdot p(S)$, where $\ell$ is sequence length, and $p_1(S)\leq d_1(S)$.
\begin{center}
	\renewcommand{\arraystretch}{1.4}
	\begin{tabular}{c|ccc}
	Case&Rank Bound&Union Bound&Probability of Success\\\hline
	$p_1\geq p_2$&$d(S)-q_0(S)+d_1(S)/2$&$k^{p(S)}(4n/\epsilon)^{d(S)-q_0(S)}(nk)^\ell$& $1-(nk\phi)^{O(k\cdot p(S))} \epsilon^{p(S)/4}$ \\[2pt]
	$p_2\geq p_1$&$p_2(S)/2$&$k^{p(S)}(nk)^\ell$&$1-(nk)^{O(k\cdot p(S))}(\ell \phi\epsilon)^{p(S)/4}$
	\end{tabular}
\end{center}
And so, setting $\epsilon = 1/(nk\phi)^{O(k)}$ suffices in  both cases for good success probability. 
We later show why $\ell= 2(d(S)-q_0(S))$ suffices. The other inequalities follow by definition.

\subsubsection{Mostly Non-Repeating Players, and a Union Bound via Bucketing}
Recall, a {\em non-repeating} player is one who plays each strategy at most once throughout the better-response sequence, including the initial strategy at the beginning of the sequence. Non-repeating players are key to showing rank bounds in the following sense: 

\paragraph{Rank Bounds through Separated Blocks.}
Let $v$ be some non-repeating player, and suppose the $\tau$-th move is $(v,\sigma)$.
Let $\tau'$ be the next move of $v$. Then we must have that $((v,\sigma)(*,*))$ entries of $L_t$ can only be nonzero for $\tau\leq t\leq \tau'$, and at least one of these entries must be nonzero. Therefore, if the submatrix consisting of columns $\{L_t:\tau<t<\tau'\}$ restricted to rows of the form $((v,\sigma)(*,*))$ can be shown to have large rank, then in some sense we may isolate this submatrix and inductively show a large rank for the rest of the matrix.

Furthermore, if there is some player $w$ which is inactive in the sequence, playing strategy $\sigma_w$ from the start, then the submatrix restricted to $((v,*)(w,\sigma_w))$ rows, and columns indexed by moves of player $v$, can again be isolated and shown to be upper-triangular.

These observations can be extended to hold simultaneously for all intervals between two consecutive moves of the non-repeating players:
Let $P_1$ be the set of non-repeating players, and let $T=\{\tau_v:v\in P_1\}$ be the set of all moves for non-repeating players. 
Suppose $T=\{t_1<t_2<\dotsc<t_m\}$, and say $t_0=0$, $t_{m+1}=|S|$.
Let $S_i$ for $i=0,\,1,\,\dotsc,\, m$ be the subsequence of $S$ between $t_i$ and $t_{i+1}$ excluding endpoints. We call these $S_i$ the {\em separated blocks} of the sequence $S$.
If the move at time $t_i$ is $(v_i,\sigma^i)$, then it can be seen that
\[
	rank(S)\geq |T| + \textstyle\sum_{i=0}^m rank\big|_{((v_i,\sigma^i)(*,*))} (S_i)
\]
where $rank\big|_{C}(S)$ is the rank of the submatrix given by $S$ restricted to entries from $C$. This is shown in Lemma~\ref{lem:rank-d-q} by applying the above observations, and sorting the blocks in increasing order.

\paragraph{Extension via Critical Subsequences.}
As argued in Lemma~\ref{lem:rank-d-q}, it is not hard to show that for any sequence $S$, $rank(S)\geq d(S)-q_0(S)$. In fact, we show that for a separable block $S_i$, $rank\big|_{((v_i,\sigma^i)(*,*))}(S_i)\geq d(S_i)-q_0(S_i)$. However, this is not enough to immediately give the desired rank bound of $d(S)-q_0(S)+d_1(S)/2$. 
For this, we introduce the notion of a {\em critical subsequence}, based on the notion of a critical block introduced in~\cite{A+17}. 
Let $\ell(S)$ denote the length of a sequence, and call a contiguous subsequence $S'\subseteq S$ {\em critical} if $\ell(S')\geq 2(d(S')-q_0(S'))$, but for every sub-subsequence $S''\subsetneq S'$, $\ell(S'')< 2(d(S'')-q_0(S''))$. 
We show in Claim~\ref{lem:crit} that for any $S$ with $|S|\geq 2nk$, $S$ must contain some critical subsequence $S'$ which satisfies $\ell(S')=2(d(S')-q_0(S'))$.

Consider, now, the rank bounds due to separated blocks applied to a critical subsequence $S'$: since every separated block $S_i$ of $S'$ is a strict subsequence of $S'$, $d(S_i)-q_0(S_i)> \ell(S_i)/2$, but $\ell(S') = d_1(S') + \sum_{S_i\text{ separated}} \ell(S_i)$, so 
\[
	rank(S') \ \geq\  
	d_1(S') + \textstyle\sum_{S_i\text{ sep.}} d(S_i)-q_0(S_i) \ > \ 
	d_1(S')/2 + \ell(S')/2
\]
Our desired rank bound follows from recalling that $\ell(S')=2d(S')-2q_0(S')$.

\paragraph{Union Bound via Bucketing.}
Working within a critical subsequence allows us to find high-rank subsequences. However, when the critical subsequence is too small, issues may arise when taking the union bound over all sequences: there are $k^n(nk)^\ell$ sequences of length $\ell$, and the rank of a sequence is at most its length. Thus,  we get success probability $1-k^n (nk)^\ell (\phi\epsilon)^{\ell}$ in the best case. 
If $\ell$ is small, the $k^n$ term dominates the probability bound, and $\epsilon$ may need to be exponentially small for good results.
The issue at hand is that the initial strategies of inactive players contribute too much to the union bound.
For the case of $p_1\geq p_2$, we separate out their effect on the potential function, and simply keep track of the effect size.
The first part of this method was introduced in~\cite{A+17}, and allows us to reduce the $k^n$ term to a $k^{p(S)}$ term, while paying in the exponent of $\epsilon$.
The idea is as follows: let $P_0$ be the inactive players, and $P_1$ be the active players. then 
\[
	\payoff(\sss) = \sum_{uv\in E}A_{uv}(\sigma_u,\sigma_v) 
	= %\underbrace{
    \sum_{u,v\in P_1}\!\! A_{uv}(\sigma_u,\sigma_v)
    %}_{\text{(a)}} 
    + 
	\sum_{u\in P_1}
    %\underbrace{
    \sum_{w\in P_0}A_{uw}(\sigma_u,\sigma_w)
    %}_{\text{(b)}} 
    + 
	%\underbrace{
    \sum_{w,w'\in P_0}\!\!A_{ww'}(\sigma_w,\sigma_{w'})
    %}_{\text{(c)}}
\]
The left terms depend only on the strategies of the active players, the right term is constant, and the middle terms can be separated into $|P_1|$ constant terms, one per active player, per strategy played. 
These constant terms may be rounded to the nearest multiple of $\epsilon$, and lie in the range $[-n,n]$. Therefore, we need only keep track of $k^{p(S)}(2n/\epsilon)^{d(S)}(nk)^\ell$ values. Unfortunately, this is shy of our goal, as we can only show rank $d(S)-q_0(S)+\Omega(p_1)$, rather than $d(S)+\Omega(p_1)$.
Shifting the sum to cancel the initial payoff of all players, however, allows us to reduce the union bound:
\[
	\payoff(\sss^t)-\payoff(\sss^0) = 
	%\underbrace{
    \sum_{u,v\in P_1}\!\! \widetilde A_{uv}(\sigma^t_u,\sigma^t_v)
    %}_{\text{(a')}} 
    + 
	\sum_{u\in P_1}%\underbrace{
    \sum_{w\in P_0}\widetilde A_{uw}(\sigma^t_u,\sigma^t_w)
    %}_{\text{(b')}} 
    + 
	%\underbrace{
    \sum_{w,w'\in P_0}\!\!\widetilde A_{ww'}(\sigma^t_w,\sigma^t_{w'})%}_{\text{Always 0}}
\]
where $\widetilde A_{uv}(\sigma,\sigma') = A_{uv}(\sigma,\sigma')-A_{uv}(\sigma^0_u,\sigma^0_v)$. This has the effect of cancelling out $q_0(S)$ distinct middle terms, getting the desired union bound.

\subsubsection{Mostly Repeating Players, and Cyclical Sums}\label{sec:cyclic}
Recall the table of bounds from Section~\ref{sec:smoothed-poly-overview}. 
The previous section showed the bounds and analyses for the first row of the table, in the case $p_1\geq p_2$. 
The previous section's analysis only works when restricted to a critical subsequence $S'$ with $p_1(S')\geq p_2(S')$. Thus, we must also restrict ourselves to a critical subsequence, where here $p_2(S')\geq p_1(S')$. 
This means we must again find ways to control the union bound terms due to inactive players.

The fundamental concept here is the notion of a {\em cyclical sum}.
These are vectors which are sums of vectors from $\mathcal L$, which all have zero entries in rows for inactive players.
Suppose player $u$ moved to strategy $i$ twice, 
and let $\tau_0$ be the first occurrence of $(u,i)$ in the BR sequence (possibly t=0), and let $\tau_1,\tau_2,\dotsc,\tau_k$ be all subsequent appearances of $u$ in the sequence, playing any strategy. Suppose $\tau_m$ is the second occurrence of $(u,i)$ in the BR sequence. 
Let $w$ be some inactive player, who is always playing strategy $\sigma_w$.
Then the sum $L_{\tau_1}+L_{\tau_2}+\dotsm+L_{\tau_m}$ cancels out all $((u,*)(w,\sigma_w))$ entries, since each one gets added in some $L_{\tau_i}$, and and removed in $L_{\tau_{i+1}}$, where $\tau_{m+1}=\tau_1$. 
% Note that we do not consider $\tau_0$ in this sum.
These sums, therefore, do not depend on the initial configurations of inactive players.

Thus, working with the cyclic sums gives a $k^{p(S)}(nk)^\ell$-sized union bound.
However, the main lemma of the common framework does not directly apply, since we are bounding the rank of these cyclic sums instead of $\mathcal L$. We use the fact that $\Pr[\ip{L_t,A}\in(0,\epsilon)\text{ for }t=\tau_1,\tau_2,\dotsc]\leq \Pr[\ip{L_{\tau_1}+L_{\tau_2}+\dotsm+L_{\tau_m},\,A}\in (0,\ell\epsilon)]$ and bound this instead.
% To bound the probability that $\ip{L_t,A}$ lie in $(0,\epsilon)$ for all $t$, it suffices to bound the probability that $\ip{L_{\tau_1}+L_{\tau_2}+\dotsm+L_{\tau_m},\,A}$ lie in $(0,\ell\epsilon)$. 

\paragraph{A Rank Bound for Cyclic Sums.} It remains to show that these cyclical sums have large rank. 
This rank bound will be applicable in the case of non-complete graphs as well as complete. 
A cyclical sum must contain a non-zero entry, as otherwise the sequence can not be improving.
This allows us to form an auxiliary digraph, where the nodes are the active players, and we add an edge from a repeating player $u$ to any other player $v$, if {\em some} cyclic sum for $u$ contains a non-zero entry with $v$. 
We show in Lemma~\ref{lem:rank-p2} that there exists a way to bi-partition the nodes of this graph such that for one of the two halves contains at least half of the repeating players, and each node in this half has an out-neighbor in the other half. 
These edges, and their associated matrix entries, allow us to form an upper-triangular sub-matrix of rank $p_2/2$, giving us the desired bound.

\subsection{Smoothed Quasipolynomial Complexity of the BRA on General Graphs}
We have shown above that for complete game graphs, the BRA terminates in polynomial time with high probability, and in expectation.
We discuss here how the cyclic-sum interpretation of Section~\ref{sec:cyclic} immediately gives quasi-polynomial smoothed complexity for arbitrary game graphs. Lemma~3.4 in~\cite{ER14} proves that any sequence of $\Omega(n)$ {\em improving} moves must contain a subsequence $S'$ with at least $\Omega(\ell(S')/\log n)$ distinct repeated moves. For $k$-\nashcoord, this implies that a sequence of $\Omega(nk)$ player-strategy pairs in a BRA sequence must contain a subsequence $S'$ with at least $\Omega(\ell(S')/\log(nk))$ distinct recurring pairs. Therefore, $p_2\geq \Omega(\ell/k\log(nk))$, since each player can only appear in $k$ distinct pairs. 
This fact, along with the discussion in Section~\ref{sec:cyclic}, allow us to show the following:
\begin{thm}\label{thm:qpoly}
Given a smoothed instance of {\em $k$-\nashcoord} with an arbitrary initial strategy profile, 
then any execution of a BR algorithm where improvements are chosen arbitrarily will converge to a PNE in at most 
% $O(\phi(nk)^{(1+\eta)k\log(nk)})$
$\phi\cdot (nk)^{O(k\log(nk))}$
% steps, with probability $1-1/(nk)^{O(\eta)}$, for any fixed $\eta>0$.
steps, with probability $1-1/poly(n,k)$.
\end{thm}
\vspace*{-1ex}

Our results also include a meta-theorem stating that, for any problem in PLS with bounded total improvement, a $w.h.p.$ smoothed complexity bound implies an expected-time bound (see Section~\ref{sec:exp}, Theorem~\ref{thm:expected-time}). 
This, together with Theorem \ref{thm:poly}, imply:

\begin{thm}\label{thm:expectation}
Given a smoothed instance of {\em $k$-\nashcoord} with an arbitrary initial strategy profile, 
then any execution of BRA where improvements are chosen arbitrarily will converge to a PNE in $O(\phi)\cdot (nk)^{O(k\log(nk))}$ steps in expectation.
Furthermore, if the game graph is complete, it converges in at most $(nk\phi)^{O(k)}$ steps in expectation.
\end{thm}

% Along with the same meta-theorem that was used to show polynomial expected-time convergence of the BRA for complete graphs, we also show

% \begin{thm}\label{thm:poly-expectation}
% Given a smoothed instance of {\em $k$-\nashcoord} with an arbitrary initial strategy profile, 
% then any execution of BRA where improvements are chosen arbitrarily will converge to a PNE in at most 
% $O(\phi)\cdot (nk)^{O(k\log(nk))}$ steps in expectation.
% \end{thm}

\subsection{Smoothness-Preserving Reduction to 1- and 2-\maxcut}\label{sec:overview-reduction}
Our second set of results analyzes {\em smoothness-preserving} reductions for the $k$-\nashcoord problem, to give alternate algorithms, and to introduce a reduction framework for smoothed problems.
The authors of this paper do not know of any prior notion of smoothness-preserving reduction.
Recall the {\em smoothness-preserving} reductions defined in Section \ref{sec:PrelSmooth}. In this section, we discuss how these reduction may be applied (Theorem \ref{thm:meta-red}), and then give an overview of how to obtain such reductions for the $k$-\nashcoord problem and prove Theorems \ref{thm:kstrat} and \ref{thm:2strat} below.
All the details and proofs are in Section \ref{sec:red}.

First, if problem $\CP$ admits a {\em strong} smoothness-preserving reduction to problem $\CQ$, and $\CQ$ has smoothed polynomial algorithm {\em w.h.p.}, then so does $\CP$. Since such reductions allow use of extra randomness, we need to ensure that it does not affect the {\em high-probability} statement by much. We do this using Markov's inequality and careful interpretation of the extra randomness. 
Second, if problem $\CP$ admits a {\em weak} smoothness-preserving reduction to {\em local-max-cut} on (in)-complete graphs, then $\CP$ has smoothed (quasi)-polynomial complexity. This crucially requires a rank-based analysis like \cite{ER14,A+17,BCC18} for local-max-cut.
These are formalized in the following result: 
\begin{thm}\label{thm:meta-red}
Suppose problem $\CQ$ has (quasi-)polynomial smoothed complexity. Then, if problem $\CP$ admits a strong smoothness-preserving reduction to $\CQ$, then $\CP$ also has (quasi-)polynomial smoothed complexity. If instead, $\CP$ admits a weak smoothness-preserving reduction to local-max-cut on an (arbitrary) complete graph, then $\CP$ again has a (quasi)polynomial smoothed complexity.
%If instead, $\CP$ admits a weak smoothness-preserving reduction to local-max-cut with (in)complete graph, then $\CP$ has a (quasi)polynomial smoothed complexity.
\end{thm}
\vspace*{-1ex}

 %FLIP algorithm for the local-max-cut is based on viewing every step as linear-combination of the input parameters, and the rank argument for these linear-combinations. 
%The proofs of \cite{ER14} and \cite{A+17} mainly show that linear combinations of the edge weights have high rank, and then directly apply Lemma~\ref{lem:praobability} to show that the improvement happens with high probability. Let $M$ be a full rank $d\times d$ matrix, and note that $\ip{\alpha,MX} = \ip{M^{\mathsf T}\alpha, X}$. We conclude then that the input to Local Max Cut need not be independent random variables with bounded density. Instead, it suffices to have the edge weights be full-rank, integer linear combinations of a random variable with bounded density, since we need $M^{\mathsf T}\alpha$ to be an integer vector for all integer $\alpha$. Therefore, a {\em weak} smoothness-preserving reduction to Local Max Cut on an (in)complete graph suffices to show smoothed (quasi-)polynomial complexity.
%\medskip

This allows us to extend smoothed efficient algorithms for one problem to others. 
Ideally we would like to show {\em strong} reductions. However, we manage to show {\em weak} reductions from $k$-\nashcoord to $1$- or $2$-\maxcut. %$(i)$ from $2$-\nashcoord to $1$-\maxcut (local-max-cut), and $(ii)$ from $k$-\nashcoord to $2$-\maxcut. 
Note that the smoothed complexity of $2$-\maxcut is not known yet, but we believe that the $2$-FLIP algorithm may admit a similar rank-based analysis as FLIP. This would imply a smoothed efficient algorithm for $k$-\nashcoord for non-constant $k$. %Next, we discuss the second reduction.

\paragraph{$k$-\nashcoord to $2$-\maxcut.}
We wish to show the following result:
\begin{thm}\label{thm:kstrat}
% There exists a {\em weak smoothness-preserving} reduction from k-\nashcoord to {\em 2-\maxcut}.
$k$-\nashcoord admits a {\em weak} smoothness-preserving reduction to 2-\maxcut.
\end{thm}
\vspace*{-1ex}

%Though this result does not immediately provide an algorithm to solve $k$-\nashcoord, we show that it does raise the following implication:
%If we can show that the 2-FLIP algorithm solves smoothed instances of \mbox{2-\maxcut} in (quasi-)polynomial time through rank-based analysis similar to FLIP, then \nashcoord has (quasi-)polynomial smoothed complexity for $k$ in the input, rather than fixed $k$.
%We leave the smoothed analysis of 2-\maxcut as an open problem.

The idea for the reduction is, given an instance of a Network Coordination Game, construct a graph whose locally optimal cuts can be mapped to strategy profiles, and cut values can be interpreted as the total payoff of the network game.
To do this, we construct a graph with $nk+2$ nodes, including two terminals $s$ and $t$, and $nk$ nodes indexed by player-strategy pairs $(u,i)$. 
All nodes are connected to $s$ and $t$, any two $(u,*)$ nodes are connected, and any two $(u,*),\,(v,*)$ nodes are connected if $u$ and $v$ share an edge in the game graph.
Thus, the cut graph is complete if and only if the game graph is.
Call a cut $S,T$ valid if $s\in S$, $t\in T$, and $S$ contains at most one $(u,*)$ for each $u$.
Any such cut naturally maps to a strategy profile as follows:
If player $u$ appears in $S$ paired with strategy $i$, then set $\sigma_u=i$. 
Otherwise, set $\sigma_u=0$, a ``dummy'' strategy with bad payoff. Call this profile $\sss(S)$.
We select edge weights such that all local-max cuts are valid, by ensuring that any non-valid cut can always be improved by removing any redundant node.
We also get that for any valid cut $S$, the total cut value is equal to $\payoff(\sss(S))$.
Furthermore, updating a player's strategy in a valid cut amounts to removing one $(u,i)$ vertex from $S$, and adding some $(u,i')$. Therefore, a unilateral deviation is equivalent to a 2-FLIP step.

Since any deviation amounts to one step of 2-FLIP, it follows that $\sss(S)$ is a PNE if and only if $S$ is a local-max-cut up to 2 flips. We also show that edge-weights are linearly independent combinations of the random inputs, and that if the network coordination game instance satisfies the smoothed-inputs condition, then the edge weights do too, proving Theorem \ref{thm:kstrat}. %achieve a weak reduction. 

\paragraph{\bf $2$-\nashcoord to $1$-\maxcut.} Here, we take a slightly different reduction, where there are $n+2$ nodes in the graph, and any $s$-$t$ cut is interpreted as follows: if $u$ is on the same side as $s$ of the cut, $\sigma_u=1$, otherwise, $\sigma_u=2$. The same analysis goes through, but now, locally max cuts up to one flip are Nash Equilibria, which provides a weak smoothness-preserving reduction from 2-\nashcoord to 1-\maxcut, showing the following result:
\begin{thm}\label{thm:2strat}
% There exists a {\em weak} smoothness-preserving reduction from 2-\nashcoord to $1$-\maxcut. 
2-\nashcoord admits a {\em weak} smoothness-preserving reduction to 1-\maxcut.
\end{thm}
\vspace*{-1ex}
The FLIP algorithm for $1$-\maxcut has smoothed quasi-polynomial running time in general \cite{ER14}, and smoothed polynomial running time if $G$ is a complete graph \cite{A+17}. A recent result \cite{BCC18} has improved the running time of the latter, so this reduction allows us to provide better bounds on the performance of BRA in the case $k=2$. These local-max-cut results, together with Theorems \ref{thm:meta-red} and \ref{thm:2strat}, give an alternate smoothed efficient algorithm for the $2$-\nashcoord problem.

\paragraph{Acknowledgment.} We would like to thank Pravesh Kothari for the insightful discussions in the initial stages of this work.
%\newpage 

% \section*{Concluding Remarks}

% \note{Write!}

\addreferencesection
\bibliographystyle{ACM-Reference-Format}
\bibliography{nashSmoothed}

\appendix
\section{Smoothed Analysis of the BR Algorithm}\label{sec:BRA}
%A reader familiar with the works of \cite{ER14} and \cite{A+17} will find most of the analysis familiar. However we draw attention to the definitions \ref{def:transition-blocks} and \ref{def:critical-blocks}, \note{and to lemmas \_\_}, as these notions are our main contributions, within the framework of the previously known results.

In this section we formally show that the {\em better-response} (BR) algorithm finds a pure NE of a $k$-strategy \netco efficiently under the standard smoothness model (defined in Section \ref{sec:PrelSmooth}). 
We show efficiency both with-high-probability ($w.h.p.$) and in expectation. 
We begin by restating the definition of the problem:

\paragraph{Smoothed $k$-\nashcoord Problem.} 
Given a $k$-strategy \netco, defined by an undirected graph $G=(V,E)$ and an $(|E|k^2)$-dimensional payoff vector $A$, where each coordinate of $A$ is an independent random variable supported on range $[-1,1]$ with density at most~$\phi$, find a pure NE (PNE) of the game as defined in Section~\ref{sec:PrelGame}. 

Since the real-valued input to the $k$-\nashcoord problem is stochastic, the running time of any algorithm will be stochastic, and the efficiency guarantees will be $w.h.p.$ or in expectation.
We define below our notation for the better-response algorithm (BRA) and its properties.
\begin{defn}
	A {\em better-response (BR) sequence} is denoted as a sequence of (player, strategy) pairs ${S=(u_1,i_1),(u_2,i_2),\dotsc}$.
    We interpret $S$ as a sequence of player moves, where on the $t$\ts{th} move, player $u_t$ changes their strategy to play strategy $i_t$.
	It is assumed that they were not already playing $i_t$.

\paragraph{Notation.} Throughout a BR sequence, the strategy profile of the players is changing. We recall the~$\sss$ notation introduced in Section~\ref{sec:notation}: denote as $\sss^0\in\kset^n$ the initial strategy profile before the sequence $S$, and let $\sss^t$ be the profile after the $t$\ts{th} move.
Hence, $\sss^t$ differs from $\sss^{t-1}$ only in the entry for $\sigma_{u_t}$, and $\sigma^t_{u_t}=i_t$.\end{defn}

To better analyze better-response sequences, we must define some parameters as follows:

\begin{defn}[Active, Inactive, Repeating, and Non-Repeating players.]
	Let $S$ be a BR sequence, then player $u$ is said to be {\em active} if it appears 
    in the sequence, and otherwise, it is termed {\em inactive}. An active player $u$ is said to be {\em repeating} if there exists 
    some strategy $i$ such that $(u,i)$ appears at least twice in $S$, or if $(u,\sigma^0_i)$ appears in $S$ at all.
    An active player which is not repeating is said to be {\em non-repeating}.

\paragraph{Notation.} We denote as $p(S)$ the total number of active players in the BR sequence $S$, $p_1(S)$ as the number of non-repeating players, and $p_2(S)$ as the number of repeating players.
When the sequence in question is clear from context, we omit the $S$ and use $p,p_1$ and $p_2$, respectively.
Furthermore, let $d(S)$ denote the number of distinct (player,strategy) moves of $S$, and let $q_0(S)$ denote the number of $(u,\sss^0_u)$ moves in $S$, that is, moves where players return to the original strategy from the start of the sequence.\end{defn}

Our proof follows the {\em common framework} outlined in Section~\ref{sec:common}.
The following lemma is the key probability bound from this framework.

\begin{lem}[\cite{Rog08}]\label{lem:probability} Let $X\in \R^d$ be a vector of $d$ independent random variables where each $X_i$ has density bounded by $\phi$. Let $\alpha_1,\,\dotsc,\,\alpha_r$ be $r$ linearly independent vectors in $\mathbb{Z}^d$.
then the joint density of $(\ip{\alpha_i,X})_{i\in [r]}$ is bounded by $\phi^r$, and for any given $b_1,b_2,\dotsc \in \mathbb{R}$ and $\epsilon>0$, %In particular, if sets $J_i\subset \mathbb{R}$ have measure $\epsilon$, then
\begin{equation*}
\Pr\Big[\textstyle\bigwedge_{i=1}^r\ip{\alpha_i,X}\in [b_i,b_i+\epsilon]\Big] \leq (\phi\epsilon)^r
\end{equation*}
\end{lem}
% In all that follows, $\sigma_0$ is the original strategies of the players before the BR sequence,
% and $\sigma_i$ is the strategies after the first $i$ moves have been made.

\vspace*{3pt}
As outlined in Section~\ref{sec:common}, the goal is to model the improvement in potential at each step by some linear combination of the random payoff values, and show that it is unlikely that {\em every} step has improvement in $(0,\epsilon)$. If  some step has improvement either larger than $\epsilon$, or less than 0, then the sequence is either not an improving sequence, or has good improvement.
To this end, we define the notion of a {\em transformation set} below.
Not only do we want large improvement for any improving sequence, we also want this to hold simultaneously for all (by taking a union bound): 
if {\em every} $BR$ sequence of length $2nk$, for arbitrary starting configuration of players, is either not improving, or has a step of size greater than $\epsilon$, then the BRA must terminate in at most 
\(
	2nk\cdot n^2/\epsilon
\)
steps, since the potential function lies between $-\binom n2$ and $\binom n2$.
We seek to show sufficiently good probability bounds such that this holds for $\epsilon = 1/(nk\phi)^{O(k)}$ for complete game graphs, and $1/(nk\phi)^{O(k\log(nk))}$ for general game graphs.
Thus, we would have that the BRA terminates in at most $(nk\phi)^{O(k)}$ steps for complete graphs, and $(nk\phi)^{O(k\log(nk))}$ in general.
We formalize this as follows:
\begin{defn}[Minimum Improvement]\label{def:Delta}
    For a fixed sequence of moves, say $S$, the change in potential as time progresses is a random variable, since the values being added and subtracted (the $A$ payoff values) are random. 
    Therefore, any sequence of moves $S$ has some probability of being a BR sequence, $i.e.$ a sequence of moves every one of which is an increase in the potential.
    We define the random variable $\Delta_N$, which captures the total increase in potential of the worst BR sequence of length exactly $N$ moves. 
    Therefore, for {\em any} fixed sequence of moves $S$ of length $N$,
    and any arbitrary initial profile $\sss^0$, either $S$ is not a BR sequence under random payoffs $A$, or performing $S$ increases the potential function by at least $\Delta_N$.
\end{defn}
As described above, we wish to show that for $N=2nk$, the probability that $\Delta_N\leq 1/(nk\phi)^{O(k)}$ is vanishingly small if the game graph is complete, and the probability that $\Delta_N\leq 1/(nk\phi)^{O(k\log(nk))}$ is vanishingly small in general, which allows us to bound the total duration of a BR sequence.

%----------------------
\begin{defn}[Transformation Set]
For any fixed BR sequence of length $\ell$, we define the set of vectors 
$\Lcal=\{L_1,\,L_2,\,\dotsc,\,L_\ell\}$, where $L_t\in \{-1,0,1\}^{(|E|\times k^2)},\forall t$.
The entries of $L_t$ are indexed by indices of payoff matrices, denoted $((u,i)(v,j))$.
The values of its entries are chosen as follows:

\begin{equation*}
L_t((u,i)(v,j))=\left\{\begin{array}{llll}
1 & \text{if: ~~}u_t\in \{u,v\}&\text{and ~~}\sigma^t_u=i&\text{and ~~}\sigma^t_v=j.\\
-1 & \text{if: ~~}u_t\in \{u,v\}&\text{and ~~}\sigma^{t-1}_u=i&\text{and ~~}\sigma^{t-1}_v=j.\\
0 & \text{otherwise.} 
\end{array} \right.
\end{equation*}
That is, every entry signifies if the corresponding payoff value remains unchanged $(0)$, or gets added $(+1)$ or removed $(-1)$ from the potential function. 
We denote this set as the \textit{\Lset} of a sequence, and each vector $L_i$ as the \textit{\Li} of the corresponding move.  
\end{defn}
The inner product $\ip{L_t,A}$ gives the change in the total payoff of the player who makes a move at time $t$, hence also gives the change in the sum of payoffs due to the $t$\ts{th} move of the entire game.
The rest of the analysis is dedicated to ensuring that the \Lset has large rank, large enough to counter a union bound over all sequences.

\subsection{Smoothed Polynomial Complexity for Complete Game Graphs}
Our analysis begins with the case where the game graph is complete.
The completeness of the game graph will allow us to single out a vertex, and use edges to that vertex to show large rank.
We wish to show that the BRA algorithm will terminate in $poly(n^k,k,\phi)$ steps with high probability.
We begin by handling a simple case: 

\medskip
\noindent\textbf{Informal Lemma.} \textit{Suppose we could guarantee that every $2nk$-length sequence had all players active, then any execution of BRA must terminate in $(nk\phi)^{O(k)}$ steps with probability $1-1/(nk)^{O(1)}$.}
\begin{proof}
	From Corollary~\ref{cor:all-active}, which we will show below, if $S$ has all $n$ players active, then the transformation set $\Lcal$ of $S$ has rank at least $\big(1-\tfrac 1n\big)\big(d(S)-q_0(S)\big)\geq n-1$. Thus, any such sequence $S$ is either $\epsilon$-improving (has improvement at least $\epsilon$) or non-improving (non-positive improvement) with probability $1-(\phi \epsilon)^{n-1}$. 
    However, we want this to hold for every sequence.
    For any fixed $\ell$, there are $(nk)^\ell$ BR sequences of length $\ell$, and $k^n$ possible initial configurations. Thus, the probability that {\em every} sequence of length $2nk$ is neither non-improving nor $\epsilon$-improving, for any initial configuration, is at most $k^n(nk)^{2nk}(\phi\epsilon)^{n-1}$.
    For this value to be vanishingly small, {\em i.e.} $1/(nk)^{O(1)}$, it suffices to set $\epsilon= 1/\phi(nk)^{O(k)}$.
    As discussed above, this implies that with probability $1-1/poly(n,k)$, the BRA will terminate in time $2n^3k/\epsilon = (nk\phi)^{O(k)}$.
\end{proof}

The above lemma, however, relies on a condition which cannot be guaranteed: not every sufficiently long sequence has $n$ active players.
In what follows, we lower bound the rank of the \Lset for sufficiently long BR sequences, and pair them with nontrivial union bounds, to get our desired results.

\subsubsection{Case I: Mostly Non-Repeating Players}

We will be showing rank bounds which depend on the $p_1$ and $p_2$ parameters defined above. 
Splitting this analysis into the cases $p_1\geq p_2$ and $p_2\geq p_1$ allows us to combine these and get bounds in terms of $p$.
We consider first the case $p_1\geq p_2$, and the following definition:

\begin{defn}[Separated Blocks]\label{def:separated}
	Let $P_1(S)$ be the set of non-repeating players in a BR sequence, and for any $u\in P_1$, let $T_u$ be the set of indices where the moving player is $u$. 
    Let $T=\bigcup_{u\in P_1}T_u$, the set of indices of all non-repeating-player moves, and suppose $T=\{t_1<t_2<\dotsm<t_m\}$
    We will show below how the $t_i$'s ``separate'' the sequence $S$. To this end, let $S_i$ for $i=0,\,1,\,\dotsc,\,m$ be the subsequences of $S$ from time $t_i$ to $t_{i+1}$ excluding boundaries, respectively, where $t_0=0$ and $t_{m+1}=|S|$. 
    Then these $S_i$'s are the {\em separated blocks} of $S$. 
    Denote their collection as $\mathbb S = \{S_0,\,S_1,\,\dotsc,\,S_m\}$.
    Furthermore, denote $|T|$ as $d_1(S)$.
\end{defn}

The following lemma allows us to take advantage of this notion of separated block, to break up the rank bounds into simpler subproblems.

\begin{lem}\label{lem:rank-d-q}
Let $S$ be a BR sequence with at least one inactive player, and let $\Lcal=\{L_1,\,L_2,\,\dotsc\}$ be its \Lset.
Then $\Lcal$ contains at least $d_1(S)+\sum_{S'\in \mathbb S}d(S')-q_0(S')$ linearly independent vectors, where $\mathbb S$ is the collection of separated blocks of $S$. 
\end{lem}
\begin{proof}
Let $w$ be some inactive player, which we have assumed exists.
Let $T=\{t_1<t_2<\dotsm<t_m\}$ be the endpoints of the separated blocks, as in Definition~\ref{def:separated} above.
For $i=0,\,1,\,\dotsc,\,m$,
let $D_i$ be the set of distinct (player,strategy) moves which occur in $S_i$, which are not {\em return moves} of $S_i$, {\em i.e.} $(u,\sss^{t_i}_u)$ moves.

For all $i$, the move at $t_i$ must be some non-repeating player of $S$, which we denote $v_i$, and call the strategy it moves to as $\sigma^i$. (Take $v_0=w$, $\sigma^0=\sss^0_w$).
For all $(u,\sigma)\in D_i$, let $\tau^i_{(u,\sigma)}$ be the time of the {\em first} occurrence of $(u,\sigma)$ in the subsequence~$S_i$, and let $H_i=\{\tau^i_{(u,\sigma)}:(u,\sigma)\in D_i\}$. 
Let $H=\bigcup_{i=0}^{|\mathbb S|-1} H_i \cup \{t_1,\,\dotsc,\,t_{|\mathbb S|-1}\}$. 
For each $t\in H$, if $t=\tau^i_{(u,\sigma)}\in H_i$ for some $i,\,u,\,\sigma$, then {\em associate} to $L_t$ the row $((u,\sigma)(v_i,\sigma^i))$. If, instead, $t=t_i$ for some $i$, then {\em associate} to $L_t$ the row $((v_i,\sigma^i)(w,\sigma^0))$.

Consider the submatrix of $\Lcal$ consisting of all columns $\{L_{t}:t\in H\}$,
sorted in ``chronological'' order, and all of their associated rows, in the same order as their respectively associated columns. We claim that this matrix is upper-triangular, and its diagonal entries are non-zero. For each column $L_t$, the diagonal entry in the submatrix is the entry for the associated row, which we have chosen to be nonzero. 
Furthermore, if $t=t_i\in H$, then $v_i$ no $((v_i,\sigma^i)(*,*))$ entry can have been non-zero, since $v_i$ is non-repeating. Thus, $L_{t_i}$ is the first column where the associated row has a nonzero entry.
If, instead, $t=\tau^i_{(u,\sigma)}\in H_i$, then the associated row $((u,\sigma)(v_i,\sigma^i))$ must have been 0 up until column $L_{t_i}$ as described above.
Furthermore, since $\tau^i_{(u,\sigma)}$ is the first occurrence of $(u,\sigma\neq \sss^{t_i}_u)$ after time $t_i$, we must have had the row $((u,\sigma)(v_i,\sigma^i))$ be 0 before the $\tau^i_{(u,\sigma)}$-th column. 

These observations imply that our $|H|\times|H|$ submatrix, with the given row-ordering, must be upper-triangular with nonzero diagonal terms. Therefore, it must be full-rank. 
Since $|H_i| = d(S_i)-q_0(S_i)$, then $|H|=d_1(S)+\sum_{S'\in \mathbb S} d(S')-q_0(S')$, and we conclude the desired bound.
\end{proof}

\begin{cor}\label{cor:all-active}
Let $S$ be a BR sequence where all players are active, and let $\Lcal=\{L_1,\,L_2,\,\dotsc\}$ be its \Lset.
Then $\Lcal$ contains at least $\left(1-\tfrac 1n\right)(d(S)-q_0(S))$ linearly independent vectors.
\end{cor}
\begin{proof}
Consider the above proof method with $|T|=0$, and $S_0=S$. Note that now, $H=D_0$.
It is still correct if some arbitrary player is chosen to be the $w$ player, and all 
$((u,\sigma)(v_0,\sigma^0))$ terms are replaced with $((u,\sigma)(w,\sss^{\tau^0_{(u,\sigma)}}))$ terms.
We must further restrict $H$ not to contain any moves of player $w$.

Suppose we choose, as our $w$ player, the player which appears the least number of times in $H$, then we suffer a $\left(\tfrac 1n\right)$-fraction loss in the size of $H$, concluding the proof.
\end{proof}

Now, the rank bound of Lemma~\ref{lem:rank-d-q} is not in a very usable form, as it requires too much structural information about the sequence $S$ to use. 
To this end, we turn to the definition of {\em critical subsequence}, closely based on the definition of a critical block in~\cite{A+17}.

\begin{defn}[Critical Subsequence]\label{def:critical-blocks} For every contiguous subsequence $B$ of $S$, let $\ell(B)$, $d(B)$, and $q_0(B)$ be length, number of distinct pairs, and number of return moves, in $B$, respectively. 
Such a subsequence is termed {\em critical} if $\ell(B)\geq 2\big(d(B)-q_0(B)\big)$, but for every $B'\subseteq B$, 
$\ell(B')< 2\big(d(B')-q_0(B')\big)$.
\end{defn}

Note that a return move for a subsequence $B$ which starts at time $t_B$ is a move $(u,\sss^{t_B}_u)$, as opposed to a $(u,\sss^0_u)$ move. 
We show here that critical subsequences always exist.

\begin{claim}\label{lem:crit}
A critical subsequence always exists in any sequence $S$ of length $2nk$.
Furthermore, if $B$ is a critical subsequence, then $\ell(B) = 2(d(B)-q_0(B))$.
\end{claim}
\begin{proof}
As there are at most $nk$ distinct player-strategy pairs possible, 
the entire sequence $S$ satisfies the relation $\ell(S)\geq 2d(S)\geq 2(d(S)-q_0(S))$. 
Conversely, for every subsequence $B$ of length 1 ($i.e.$ a single move), $d(B)=1,q_0(B)=0 \Rightarrow 1=\ell(B)<2(d(B)-q_0(B))=2$.
Thus, it suffices to take an inclusion-minimal subsequence which satisfies $\ell(B)\geq 2(d(B)-q_0(B))$ and obtain a critical subsequence. 

It remains to show that for $B$ critical $\ell(B) = 2d(B)-2q_0(B)$. Suppose not, then it is strictly larger. Let $B'$ be obtained from $B$ by dropping the last column.
Then, 
\[
	\ell(B') \ = \ 
    \ell(B)-1 \ \geq \ 
    2d(B) - 2 q_0(B) + 1 - 1
\]
Now, we claim $d(B)-q_0(B)\geq d(B')-q_0(B')$. Clearly $d(B)-1\leq d(B')\leq d(B)$, and $q_0(B)-1\leq q_0(B') \leq q_0(B)$. However, if $q_0(B')= q_0(B)-1$, then we must also have $d(B')=d(B)-1$. Thus, in all cases, $d(B)-q_0(B)\geq d(B')-q_0(B')$. 
This implies $\ell(B')\geq 2(d(B')-q_0(B'))$, contradicting the criticality of $B$.
\end{proof}

The tight bound $\ell(B)=2d(B)-q_0(B)$ is the final key in proving the main rank lemma from this section, below. Since we prove that any critical subsequence has good rank, and any sequence has a critical subsequence, then this shows that any length-$2nk$ sequence has a high-rank subsequence.

\begin{lem}\label{lem:crit-rank}
Let $S$ be a BR sequence of length $2nk$ which has at least one inactive player.
Let $B$ be some critical subsequence of $S$ and let $\Lcal$ be $B$'s \Lset. Then $\Lcal$ contains at least $\tfrac12d_1(B) + d(B)-q_0(B)$ linearly independent vectors.
\end{lem}
\begin{proof}
	Since $S$ has an inactive player, then so must $B$. 
	Therefore, Lemma~\ref{lem:rank-d-q} applies. 
    Recall, Lemma~\ref{lem:rank-d-q} shows that $\Lcal$ contains at least $d_1(B)+\sum_{S'\in \mathbb S(B)} d(S')-q_0(S')$ linearly independent vectors.
    If $p_1(B)=d_1(B)=0$, then we are done.
    Otherwise, since $B$ is critical, then for all $S'\in \mathbb S(B)$, $\ell(S')<2(d(S')-q_0(S'))$.
   Hence,
   \[
   	rank(\Lcal) \ \geq\  d_1(B)+\sum_{S'\in \mathbb S(B)}d(S')-q_0(S') \ > \ 
    \tfrac12 d_1(B) + \tfrac12 d_1(B) + \sum_{S'\in \mathbb S(B)} \tfrac12 \ell(S')
   \]
   However, $\ell(B) = d_1(B) + \sum_{S'\in \mathbb S(B)} \ell(S')$, and so this implies $rank(\Lcal)\geq \tfrac12 d_1(B)+\tfrac12 \ell(B)$. 
   By criticality and Claim~\ref{lem:crit}, $\ell(B)\geq 2(d(B)-q_0(B))$, giving us our desired bound.
\end{proof}

Since we have bounded the total increase on critical subsequences, though, rather than long sequences, we introduce the following notation:
\begin{defn}\label{defn:Delta-p}
Recall the notation $\Delta(\ell)$ from Definition \ref{def:Delta-l}. We similarly define $\Delta'(p)$ as the minimum total increase due to any critical subsequence with exactly $p$ active players, where the initial strategy profile is arbitrary. 
We also denote as $\overline{\Delta}(p)$ the minimum total increase taken over critical subsequences where $p_1\geq p_2$, and $\underline \Delta(p)$, the converse. Thus, $\Delta'(p) = \min\{\overline \Delta(p),\,\underline\Delta(p)\}$.
\end{defn}

Our first main result is to show that the probability of $\overline\Delta(p)$ being small is vanishing. 
It does not suffice, however, to follow the same structure as the ``informal lemma'' above: critical subsequences may be very short, and for any sequence $S$, its rank can not be more than $\ell(S)$. 
A probability bound of the form $k^n(nk)^\ell(\phi\epsilon)^\ell$
will require $\epsilon$ far too small if $\ell\ll n$, since the $k^n$ term will dominate.
The following proof technique illustrates that no $k^n$ term is needed in a union bound when $p$ is small.

\begin{thm}\label{thm:prob-p1}
	\(
\Pr\big[\:\overline \Delta(p)\in (0,\epsilon)\big] \leq \left(\left(20\phi^2n^3k^3\right)^k\epsilon^{1/4}\right)^p.
\)
\end{thm}
\begin{proof}
	We have shown in Lemma~\ref{lem:crit-rank} that a critical subsequence has high rank. However, as discussed above, this does not suffice to counteract simple union bounds.    
    Recall the potential function $\payoff(\bm \sigma)$ from \eqref{eq:pot}, which represents the sum of the payoffs on all game edges at strategy $\bm \sigma$. 
    For simplicity of notation, let $H(t)$ denote $\payoff(\sss^t)$. 
    Note that analyzing the change in $H(t)$ through any BR sequence is unaffected by shifting $H$ by a constant. We will the initial payoff as this constant: let $H'(t):= H(t)-H(0)$.
    
    Our goal is to bound the rate of change of $H'(t)$ over all possible BR sequence which form critical subsequences. 
    Let $P$ be the set of all active players, and $Q$, the set of inactive players. 
    Note that 
    \begin{align*}
	H'(t)&= \sum_{u=1}^{n-1} \sum_{v=u+1}^n A((u,\sss^t_u)(v,\sss^t_v))-A((u,\sss^0_u)(v,\sss^0_v))
    \end{align*}
    For simplicity of notation, denote $\widetilde A((u,\sigma_u)(v,\sigma_v)) = A((u,\sigma_u)(v,\sigma_v))-A((u,\sss^0_u)(v,\sss^0_v))$. Then
    \begin{align*}
    H'(t):= \sum_{u,v\in P} \widetilde A((u,\sss^t_u)(v,\sss^t_v)) + \sum_{w,w'\in Q} \widetilde A((w,\sss^t_w)(w',\sss^t_{w'})) + \sum_{u\in P} \sum_{w\in Q} \widetilde A((u,\sss^t_u)(w,\sss^t_{w})) 
    \end{align*}
   Now, for $w\in Q$, $\sss^t_w=\sss^0_w$, so the second term is 0.
   Furthermore, the inner-sum of the 3rd terms are in fact constants which depend on the strategy of the active player. Thus, define $C(u,\sigma):= \sum_{w\in Q} \widetilde A((u,\sigma)(w,\sss^0_w))$. Then the above sum can be expressed as
   \begin{align*}
    H'(t):= \sum_{u,v\in P} \widetilde A((u,\sss^t_u)(v,\sss^t_v)) + 0 + \sum_{u\in P} C(u,\sss^t_u) 
    \end{align*}
    Also, note that $C(u,\sss^0_u)=0$, since the $\widetilde A$ terms cancel.
    This reduction to $C$ terms is the key to our analysis: 
    To bound $\overline\Delta$, we must bound $H'(t)-H'(t-1)$ for all $t\geq 1$,
    and these values, in turn, depend only on the vector $A$, the initial strategies of the active players, the improving sequence, and the $C(u,\sigma)$ values for each of the pairs $(u,\sigma)$ which appear in the sequence. As noted, $C(u,\sss^0_u)=0$ for all $u$, so we need not consider them for initial strategies. 
    
    We can not, however, enumerate all the possible $C$ values in our union bound. 
    Note, however, that approximating the $C$ values simply approximated the $H$ values. Thus, we may round the $C$ values to the nearest multiple of $\epsilon$, as was first introduced in~\cite{A+17}. 
    Let $C'(u,\sigma)$ be the nearest multiple of $\epsilon$ to $C(u,\sigma)$. 
    Since $C(u,\sigma)\in [-n,n]$ for all $u$ and $\sigma$, then there are $2n/\epsilon$ possible choices for $C'(u,\sigma)$. 
    Furthermore, for any time $t$, note that $H(t)-H(t-1)$ depends only on two $C$ terms, namely $C(u,\sss^{t}_u)$ and $C(u,\sss^{t-1}_u)$. Thus, replacing these with $C'$ terms modifies the value of $H(t)-H'(t)$ by at most $2\epsilon/2$, and so the event $H'(t)-H'(t-1)\in (0,\epsilon)$ is less likely than $\widetilde H'(t)-\widetilde H'(t-1)\in (-\epsilon,2\epsilon)$, where $\widetilde H'$ is simply the approximation due to replacing $C$ with $C'$.
    
    It remains, then, to apply Lemma~\ref{lem:probability}: 
    for any one critical subsequence $S$ on $p$ players, if $p_1\geq p_2$, then by Lemma~\ref{lem:crit-rank}, the improvement of each step of the {\em approximate} potential along the sequence will lie in $(-\epsilon,2\epsilon)$ with probability $(3\phi\epsilon)^{d(S)-q_0(S)+p(S)/4}$. 
    Taking a union bound over {\em all} approximated sequences, this event holds with probability $k^{p(S)}(nk)^{\ell(S)}(2n/\epsilon)^{d(S)-q_0(S)}(3\phi\epsilon)^{d(S)-q_0(S)+p(S)/4}$, since there are only $d(S)-q_0(S)$ $C$ values to approximate, and the rest of the value depends only on the initial configurations of the active players.
    Thus, noting that $d(S)-q_0(S)\leq k\cdot p(S)$, and since $S$ is critical, $\ell(S)\leq 2d(S)-2q_0(S)\leq 2kp(S)$, we have
    \begin{align*}
    	\Pr\big[\:\overline\Delta \in (0,\epsilon)\big]&\leq k^{p(S)}(nk)^{\ell(S)}(2n/\epsilon)^{d(S)-q_0(S)}(3\phi\epsilon)^{d(S)-q_0(S)+p(S)/4}\\
        &\leq 20^{k\cdot p(S)}(nk\phi)^{2k\cdot p(S)}(nk)^{k\cdot p(S)} \epsilon^{p(S)/4}\\
        &= \left((20 n^3k^3\phi^2)^{k}\epsilon^{1/4}\right)^{p(S)}
   \end{align*}
   as desired.
\end{proof}

\subsubsection{Case II: Mostly Repeating Players}\label{sec:poly-case-p2}
We have shown in the previous section how to bound the probability of $\overline \Delta$ being small. In this case, we handle $\underline \Delta$, that is, the minimum improvement of critical subsequences when $p_2\geq p_1$.
The proof method in this case is very different from the converse case, but as above, we must still use work around the $k^n$ term in the standard union bound. This is done, in this case, by combining the columns of $\Lcal$ into vectors which have zero entries for all payoff values with inactive players, and bounding their rank instead.

%Recall that, for any BR sequence, we define a set of vectors $\mathcal{L}=\{L_1,L_2,..\}$ that capture the change in the potential function with every move. These vectors are linear combinations of the $\mathcal A$ payoffs, and formally defined in Section \ref{overview:FLIP}.

%\note{want the entire definition of $\mathcal{L}$ repeated or just recall? if just recall, include the comment above and remove the definition}
%-----------------------------------

\begin{lem}\label{lem:rank-p2}
Let $S$ be a BR sequence, and let $\Lcal=\{L_1,\,L_2,\,\dotsc\}$ be its \Lset.
Then the span of $\Lcal$ contains at least $p_2(S)/2$ linearly independent vectors $V_1,\,\dotsc,\,V_{p_2/2}$,
such that for all~$j$,
\begin{enumerate} 
\item[(i)] the vector $V_j$ is a 0-1 combination of the $L_i$'s, and
\item[(ii)] the value $\ip{V_j,A}$ does not depend on the strategies of the inactive players. 
\end{enumerate}
\end{lem}
\begin{proof}
	Fix a repeating player $u$, and denote one of its repeating strategies as $\sigma$. 
    Let $t_0$ be the index of the first occurrence of $(u,\sigma)$ in $S$, setting $t_0=0$ if $i=\sss^0_u$.
    Let $t_0<t_1<t_2<\dotsm$ be all occurrences of player $u$ in the sequence $S$ starting with $t_0$, and suppose $t_\sigma$ is the index of the second occurrence of $(u,\sigma)$.
    Formally, if $S=(u_1,i_1),(u_2,i_2),\dotsm$, then $u_{t_j}=u$ for all $j$,
    $i_{t_0}=i_{t_s}=\sigma$, and $i_{t_j}\neq \sigma$ for all $0< j<s$. Consider, then, the vector $V(u) = \sum_{j=1}^s L_{t_j}$. 
    Vector $V(u)$ satisfies condition $(i)$ by construction. We will show that it satisfies $(ii)$, and that
    at least $p_2/2$ of them must be linearly independent.
    
     Let $w$ be any inactive player. For simplicity of notation, denote the strategy $i_{t_j}$ as $i_j$.
    Consider the inner product $\ip{V(u),A}$ when restricted to the rows indexed by $((u,*)(w,*))$:\\[-20pt]
    \begin{align*}
    	\ip{V(u),A}_{\big|((u,*)(w,*))}&=\sum_{j=1}^s A((u,{i_j})(w,\sss^0_w)) - A((u,{i_{j-1}})(w,\sss^0_w))\\
        &= A((u,{i_s})(w,\sss^0_w))-A((u,{i_0})(w,\sss^0_w))\\
        &=0
    \end{align*}
    Therefore, the inner product $\ip{V(u),A}$ is independent of the value of $\sss^0_w$.
    Since this holds for all inactive $w$, we have proved part $(ii)$. 
    
    Now, it suffices to argue that some collection of $p_2/2$ many $V$ vectors are linearly independent.
    We begin by constructing an auxiliary directed graph $G'=(V, E')$,
    where $V$ is the set of players, and $E'$ will be defined as follows:
    let $u$ be some repeating player, and define $V(u)$ as above, for the repeat move $(u,\sigma)$. 
    The $V(u)$ vector can not be entirely 0, as this would imply that $\sss^{t_s}$
    and $\sss^{t_0}$ are the same, and so the sequence could not have been strictly improving.
    Then, for {\em every} player $w$ such that $V(u)$ for $(u,\sigma)$ has a non-zero $((u,\sigma)(w,*))$ entry, add the edge $(u,w)$ to $E'$.
    
    Consider the following procedure: pick an arbitrary vertex $r_1\in P_2$, and let $T_1$ be the BFS arborescence rooted at $r_1$ which spans all nodes reachable from $r_1$. 
    Then delete $V(T_1)$ from $G'$ and repeat, picking an arbitrary root vertex $r_2\in P_2\setminus V(T_1)$, and get the arborescence $T_2$ on everything reachable from $r_2$. 
%     {\em If, when picking roots, there is a node with in-degree zero (source) and non-zero out-degree, then choose that node as the root.}
    We may continue this until every vertex of $P_2$ is covered by some arborescence. 
%     If at any point, there is some node $r_i\in P_2$ from the remaining graph which is a source of the remaining graph, then it must have been a source of the original graph, since it can not have been reachable from any removed nodes.
%     Therefore, if we are picking a root when there are no sources left, and it happens to have out-degree zero in the remaining nodes, we may, instead, choose one of its in-neighbours as the root. 
    For each $i=1,2,\dotsc$, let $T_i^0$ and $T_i^1$ be the set of nodes of $T_i$ which are of even or odd distance from $r$ along $T_i$, respectively. 
    Let $P'_i$ be the larger of $V(T_i^0)\cap P_2$ and $V(T_i^1)\cap P_2$, and $P_2':= \bigcup_{i=1}^\infty P_i'$.
    
    We must have that $|P_2'|\geq |P_2|/2  = p_2/2$. We wish to show that the collection $\mathcal V:=\{V(u):u\in P_2'\}$ is independent. 
    Every $u\in P_2'$ must have some out-neighbour $w$. If $u$ was not a leaf of the arborescence it was selected in, then it must have some out-neighbour along the arborescence, and we may choose this neighbour. This out-neighbour can not also be in $P_2'$. In this case, $V(u)$ will be the only vector from $\mathcal V$ to contain a non-zero $((u,*)(w,*))$ entry, since $w$ was not taken in $P_2'$.
    If, instead, $u$ was a leaf of its arborescence, then its out-neighbours must be in previously constructed arborescences. Let $w$ be any such neigbour, then $V(w)$ can not contain a non-zero $((u,*)(w,*))$ entry, as otherwise $u$ would have been in the other arborescence. Therefore, $V(u)$ is the only vector in $\mathcal V$ to contain a nonzero $((u,*)(w,*))$ entry.
Thus, $\mathcal V$ must contain a $|\mathcal V|\times |\mathcal V|$ diagonal submatrix, and therefore has rank at least $|\mathcal V|\geq p_2/2$, as desired.
\end{proof}

Therefore, we have shown that the transformation set of any sequence must have large rank if $p_2\geq p_1$. However, the vectors of the transformation set depend on the strategies of inactive players, whereas these $V$ vectors do not. 
We will show that this issue can be circumvented, and prove our second main result.
Recall that $\underline\Delta(p)$ is the (random) minimum improvement of any critical subsequence $S$ with $p$ players and $p_2(S)\geq p_1(S)$.

% In Section \ref{sec:PrelGame} we argued that a BR sequence improves the potential $\payoff$ function \eqref{eq:pot} of the game in every step. The next lemma bounds the probability of this improvement being very small. 

\medskip
\begin{thm}\label{thm:prob-p2}
$\Pr\big[\underline\Delta(p)\in(0,\epsilon)]\leq\left(2(nk)^{2k}k^{5/4}(n\phi\epsilon)^{1/4}\right)^p$.
% The probability that a BR sequence $S$ of length $\ell$ has a total improvement in potential at most $\epsilon$ is at most $(\ell\phi\epsilon)^{p_2(S)/2}$, where $p_2(S)$ is the number of repeating players in $S$.
\end{thm}
\begin{proof}
Let $\Lcal$ be the transformation set of some critical subsequence $S$, and $\mathcal V$ be a collection of $p_2(S)/2$ independent $V$ vectors from Lemma~\ref{lem:rank-p2} applied to $S$. 
If $V\in \mathcal V$ is given by $V=\sum_{i=1}^m L_{t_i}$ for some collection of indices $t_i$, then we have that $\Pr[\bigwedge_{i=1}^m\ip{L_{t_i},A}\in(0,\epsilon)]\leq \Pr[\ip{V,A}\in (0,m\epsilon)]$.
Since $m\leq \ell$, then taking the collection of all $V$ vectors and applying Lemma~\ref{lem:probability}, we have
\begin{align*}
	\Pr\left[\textstyle\bigwedge_{t=1}^{\ell(S)}\ip{L_{t},A}\in(0,\epsilon)\right]&\leq 
    \Pr\left[\textstyle\bigwedge_{V\in \mathcal V}\ip{V,A}\in (0,\ell\epsilon)\right]
    \leq (\ell\phi\epsilon)^{p_2(S)/2}
\end{align*}
Note that to construct the $V$ vectors, it suffices to have the initial strategies of the active players, and the BR sequence. Thus, there are at most $k^{p(S)}(nk)^{\ell(S)}$ possible collections $\mathcal V$.
Since we are restricting ourselves to critical subsequences, we have $\ell(S)=2(d(S)-q_0(S))\leq k\cdot p(S)$. 
Therefore, we have
\begin{align*}
	\Pr\big[\underline\Delta(p)\in(0,\epsilon)]&\leq k^{p(S)}(nk)^{\ell(S)}(\ell\phi\epsilon)^{p_2(S)/2}\\
    &\leq n^{2k\cdot p(S)}k^{(2k+1)p(S)}(2kp)^{p(S)/4}(\phi\epsilon)^{p(S)/4}\\
    &\leq \left(2(nk)^{2k} k^{5/4}(n\phi\epsilon)^{1/4}\right)^{p(S)}
\end{align*}
as desired.
\end{proof}

\subsubsection{Combining Both Cases}

We have shown above that $\overline\Delta(p)$ and $\underline\Delta(p)$ have vanishing probability of lying in $(0,\epsilon)$. 
In this section, we use these results to show that the BRA will terminate in time polynomial in $n^k$, $k$ and $\phi$, with high probability, when the game graph is complete. 
The following lemma combines our two previous results:

\begin{lem}\label{poly:sequence-of-length-2nk}
Given an improving sequence of length $2nk$, the minimum improvement after performing all moves in the sequence is at least $\epsilon=(nk\phi)^{-O(k)}$ with probability $1-1/poly(n,k)$. 
\end{lem}
\begin{proof}
We will perform a case analysis based on the values of $p_1(S)$ and $p_2(S)$, with cases for $p(S)=n$, $p(S)<n$ and $p_1(S)\geq p_2(S)$, and $p_2(S)\geq p_1(S)$.

If $p(S)=n$, we apply the rank bound of Corollary~\ref{cor:all-active} and take a union bound over all initial strategy profiles, and all possible sequences to get
\begin{align*}\label{poly-delta-fixed-n}
\Pr[\Delta'(n)\in (0,\epsilon)] &\leq k^n(nk)^{2nk} (\phi\epsilon)^{n-1}\\
&\leq \left(k^{3k}n^{2k}\phi\epsilon\right)^n\Big/\phi\epsilon
\end{align*}

This union bound over-counts the number of sequences with $p(S)=n$, but this isn't a problem.
Setting $\epsilon = \phi^{-1}\left(n^2k^3\right)^{-2k}$ gives, for $n$ sufficiently large, 
\(
\Pr[\Delta'(n)\in (0,\epsilon)]\leq\left(\frac 1{n^2k^3}\right)^{n}.
\)

In the converse case, we combine Theorems~\ref{thm:prob-p1} and~\ref{thm:prob-p2}, then take a union bound over all possible values of $p$ to bound the probability for any sequence of the given length.
As defined previously, $\Delta'(p) = \min\{\overline\Delta(p),\,\underline\Delta(p)\}$ and so, 
\begin{align}
	\Pr[\Delta'(p)\in (0,\epsilon)] \ \leq \ 
    \left(
    	(20\phi^2n^3k^3)^k\epsilon^{1/4}
    \right)^p +
    \left(
    	2(\phi\epsilon)^{1/4}n^{2k+1/4}k^{2k+5/4}
    \right)^p
    \ \leq\ 2\left(
    	(20\phi^2n^3k^3)^k\epsilon^{1/4}
    \right)^p
\end{align}\label{poly-delta-fixed-p}
Since any sequence of length $2nk$ must contain a critical subsequence, it suffices to set $\epsilon=\left(20\phi^2n^3k^3\right)^{-4k-4}$, and taking the union bound over all choices of $p$, we get
\begin{align*}
	\Pr[\Delta'\in(0,\epsilon)]&\leq 
    \sum_{p=1}^n \left(20\phi^2n^3k^4\right)^{-p}
    \ \leq\ \frac{1}{(20\phi^2n^3k^4)-1}
\end{align*}
Combining the two cases of $p=n$ and $p<n$ gives us our desired result.
\end{proof}

It remains to conclude that any execution of the BRA will find a PNE in polynomial time with high probability.
\begin{thm}\label{thm:whp-poly}
Given a smoothed instance of {\em $k$-\nashcoord} on a complete game graph, and with an arbitrary initial strategy profile, 
then any execution of BRA where improvements are chosen arbitrarily will converge to a PNE in at most 
% $(nk\phi)^{O(k+\eta)}$ steps, 
% with probability \mbox{$1-1/(nk\phi)^{O(\eta)}$}, for any fixed $\eta>0$.
$(nk\phi)^{O(k)}$ steps, with probability \mbox{$1-1/poly(n,k,\phi)$}.
% Given an instance of smoothed $k$-\nashcoord with complete graph, the BR algorithm converges to a Pure Nash equilibrium in time polynomial in $\phi$ and $(nk)^k$ with high probability.
% More precisely, for any constant $\eta>0$, the total runtime is $(nk\phi)^{O(k+\eta)}$ with probability at least $1-1/(nk\phi)^{O(\eta)}$.
\end{thm}
\begin{proof}
Lemma \ref{poly:sequence-of-length-2nk} directly implies the theorem. 
As outlined in the ``common framework'' (Section~\ref{sec:common}), we begin by partitioning the BRA sequence into blocks of length $2nk$.
Each such block must contain a critical subsequences, and therefore with probability $1-1/poly(n,k,\phi)$ every block in the partition increases the potential by at least $\epsilon=(nk\phi)^{-O(k)}$
Since the total improvement is at most $2n^2$, since there are only $\binom n2$ games, this implies that the BR algorithm can only make at most $2n^2(nk\phi)^{O(k)}$ moves. 
Since making one move takes time polynomial in $n$ and $k$, we are done.
\end{proof}

\subsection{Smoothed Quasi-polynomial Complexity for Arbitrary Graphs} \label{sec:quasipoly}
In this section we show the quasi-polynomial running time when the game graph $G$ is incomplete, and thus prove Theorem \ref{thm:qpoly}.
The analysis mostly uses the lemmas from Section~\ref{sec:poly-case-p2}, paired with the following definition and lemma, from~\cite{ER14}:

\begin{defn}\label{def:Delta-l}
Recall the random variable $\Delta$ from Definition~\ref{def:Delta}.
Call a sequence of length $\ell$ log-repeating if it contains at least $\ell/(5\log(nk))$ repeating moves (pairs). 
We denote as $\Delta(\ell)$ the minimum total potential-improvement after any log-repeating BR sequence of length exactly $\ell$.
\end{defn}

\begin{lem}[From \cite{ER14}, Lemma 3.4]\label{lem:min-repeating-pairs}
Let $\Delta_N$ and $\Delta(\ell)$ be as above. 
Then $\Delta_{5nk}:= \min_{1\leq \ell\leq 5nk}\Delta(\ell)$
%
%OR
%thm:in-text-qpoly, 
%Any BR sequence $S$ of length $5nk$ has a contiguous subsequence $S'$ of length $l\le 5nk$ such that at least $\ell/(5\log(nk))$ many distinct moves appear at least twice in $S'$.
\end{lem}

The proof of the above lemma proves that any sequence on $5nk$ pairs must contain some contiguous sub-sequence which is log-repeating. 
Thus, for the remainder of the analysis, it suffices to bound $\Delta(\ell)$.
%\note{Comment on how they just show that any sufficiently long sequence has a good subsequence? I'm not sure we can just black-box their result}
Since a sequence captured by $\Delta(\ell)$ must have at least $\ell/(5\log(nk))$ repeated terms, it must have $p_2\geq \ell/(5k\log(nk))$. Therefore, as we have shown in the proof of Theorem~\ref{thm:prob-p2}, we have $\Pr[\Delta(\ell)\in (0,\epsilon)]\leq k^{\ell}(nk)^\ell (\ell\phi\epsilon)^{\ell/10k\log(nk)}$. 
It suffices, then to simply take the union bound over all possible values of $\ell$.

\begin{thm}\label{thm:in-text-qpoly}
Given a smoothed instance of {\em $k$-\nashcoord} with an arbitrary initial strategy profile, 
then any execution of a BR algorithm where improvements are chosen arbitrarily will converge to a PNE in at most 
$\phi\cdot (nk)^{O(k\log(nk))}$
steps, with probability $1-1/poly(n,k)$.
\end{thm}
\begin{proof}As discussed above, 
\begin{align}\label{quasi-delta-fixed}
Pr[\Delta(\ell)\in (0,\epsilon)]&\leq k^{\ell}(nk)^{\ell}(\ell\phi\epsilon)^{\ell/10k\log(nk)}\nonumber \\
&\leq \left(k^2n(5nk\phi\epsilon)^{1/(10k\log(nk))}\right)^{\ell}
 & (\ell\leq 5nk)\nonumber\\
&\leq \left(2k^3n^2 %\phi^{1/(10k\log(nk))}
(\phi\epsilon)^{1/(10k\log(nk))}\right)^\ell. &(5^{1/10}\leq 2)
\end{align}
Setting $\epsilon = \phi^{-1}(2n^2k^3)^{-2\cdot 10k\log(nk)}$, 
this gives
\[
	Pr[\Delta(\ell)\in (0,\epsilon)]
    \leq\left(\frac{1}{2n^2k^3}\right)^\ell
\]

Let $\Delta_{5nk}$ be the improvement in potential in any length $5nk$ BR sequence. Then using Lemma \ref{lem:min-repeating-pairs}, and taking the union bound over all choices of $\ell$, we have,

\[\Pr[\Delta_{5nk}\in (0,\epsilon)]\leq 
\sum_{\ell=1}^{5nk}\Pr[\Delta(\ell)\in(0,\epsilon)]\le 
\sum_{\ell=1}^{5nk} 
(2n^2k^3)^{-\ell}\leq \frac{(2n^2k^3)^{-1}}{1-(2n^2k^3)^{-1}}=\frac{1}{2n^2k^3 -1}
\]

% Here, the second inequality follows by approximating $\ell\leq 5nk$ in the base, and number of active players $p\leq n$. We fix $\epsilon=(2n^2k^3)^{-(10klog(nk)+\eta)}\phi^{-1}$ to have the above probability value vanish with increasing $\ell$.

% We take a union bound over all possible lengths $\ell$ to establish that there is some sequence of length at most $5nk$ that has improvement at least $\epsilon$ with high probability. 

% By denoting the improvement of the sequence of length $\ell$ as $\Delta(S)$, we have:
% \begin{equation*}
% Pr[\Delta(S)\in (0,\epsilon)] \leq \sum_{\ell=1}^{5nk} (2k^3n^2\phi^{1/(10klog(nk))}\epsilon^{1/(10klog(nk))})^\ell \leq \frac{(2n^2k^3)^{-\eta}}{1-(2n^2k^3)^{-\eta}}=\frac{1}{(2n^2k^3)^\eta -1}.
% \end{equation*}

% By setting $\eta$ to an arbitrarily large constant, this probablity can be reduced to any given $\alpha = 1/(nk)^{O(1)}$.
% Note that, the above probability is over the draw of the payoff vector $A$ for the game. 
 Hence, with probability $1-1/poly(n,k)$ (over the draw of payoff vector $A$), {\em all} BR sequences of length $5nk$ will have total improvement at least $\epsilon$.
In that case, any execution of BR algorithm makes an improvement of at least $\epsilon$ every $5nk$ moves.
Since the total improvement is at most $2n^2$, we conclude that the total number of steps is at most $5nk\cdot 2n^2/\epsilon = 10n^3k(2n^2k^3)^{20k\log(nk)}\cdot \phi=
\phi\cdot (nk)^{O(k\log(nk))}$, and this occurs with probability $1-1/poly(n,k)$.
% by $O((nk)^\eta)$. To keep this quasi-polynomial, we require $\eta=O(k\log(nk))$. 
% This bounds $\alpha$ to at most $1/((2n^2k^3)^{O(k\log(nk))}-1)=1/O((nk)^{k\log(nk)})$. 
\end{proof}

This completes our analysis of the smoothed performance of BRA for finding pure Nash equilibria in network coordination games. 
In the next section, we show that this result indeed holds in expectation, and then go on to show a notion of smoothness-preserving reduction which allows us to prove alternative, conditional, algorithms for this problem.

\section{Expected Smoothed Time Complexity}\label{sec:exp}
The analysis in the previous section establishes smoothed complexity of \netcos with respect to the with high probability notion. Another aspect of smoothed analysis is to analyze the expected time of completion of the algorithm. In this section, we provide a theorem to obtain expected time results from the with high probability bounds. The results are presented in a general form to allow application to any problem in PLS that has a bounded total improvement in potential value.

\begin{thm}\label{thm:expected-time}
Given a PLS problem with input size $N$, potential function range $[-N^{r_1}, N^{r_2}]$, and a local-search algorithm $\CA$ to solve it, let $d$ be the number of distinct choices the algorithm has in each step and let $\Lambda$ be the total size of the search space of the algorithm. For an instance $I$ drawn at random with maximum density $\phi$, 
suppose the probability that any length-$N^\beta$ sequence of improving moves of $\CA$ results in total improvement in the potential value at most~$\epsilon$, is at most $\sum_{q=1}^{N^\beta}((\phi N)^{f(N)}(\phi\epsilon)^{1/g(N)})^q$. 
% Here $f(N)$ and $g(N)$ are any functions in $N$. 
Then the expected running-time of the algorithm is $O(N^{\beta+r}\cdot g(N)\cdot (\phi N)^{f(N)g(N)}\cdot \ln\Lambda)$.
Here, $f(N)$ and $g(N)$ are functions of $N$.
\end{thm}
\begin{proof}
The proof is from \cite{ER14}. As we have stated it in a more general form, the analysis is included for completeness.

The maximum improvement possible before $\CA$ terminates is the maximum change in the potential function value, given by $N^{r_2}+N^{r_1}$. 
For any integer $t\geq 1$, if the algorithm requires more than $t$ steps to terminate, then there must exist some subsequence of length $N^{\beta}$ that results in an improvement in the potential value of less than $N^{\beta}(N^{r_2}+N^{r_1})/t\leq 2N^{\beta+{\max\{r_2,r_1\}}}/t$. We denote $r:=\max\{r_1,r_2\}$.

We define a random variable $T$ as the number of steps $\CA$ requires to terminate. Using the notation $\Delta(N^\beta)$ to denote the minimum total improvement in a length-$N^\beta$ sequence of the algorithm $\CA$, this gives the probability of $\CA$ running for more than $t$ steps as:
\begin{equation*}
\Pr[T\geq t]\leq Pr[\Delta(N^\beta)\in (0,N^{r+\beta}/t)] \leq \sum\limits_{q=1}^{N^\beta} \left((\phi N)^{f(N)}\left(\phi\cdot \frac{N^{\beta+r}}{t}\right)^{1/g(N)}\right)^q.
\end{equation*}

We define $t=\gamma i$, for $\gamma=(\phi N)^{f(N)g(N)}(\phi N^{r+\beta})=\phi^{f(N)g(N)+1}N^{f(N)g(N)+\beta+r}$, and compute the probability of $T\geq \gamma i$ for any integer $i$:
\begin{equation*}
\Pr[T\geq \gamma i]\leq \sum\limits_{q=1}^{N^\beta} \left((\phi N)^{f(N)}\left(\phi\cdot\frac{N^{r+\beta}}{\gamma i}\right)^{1/g(N)}\right)^q \leq \sum\limits_{q=1}^\infty \left(\frac{1}{i}\right)^{q/g(N)} 
\leq g(N)\sum_{q'=1}^\infty \left(\frac 1i\right)^{q'}\leq 
\frac{g(N)}{i-1}.
\end{equation*}

We now sum over all values of $t$, by using that $\Pr[T\geq t]\leq \Pr[T\geq t\cdot\lceil t/\gamma \rceil]$, and compute the expected time steps as:
\begin{align*}
 \mathbb E[T]&=\sum_{t=1}^\Lambda \Pr[T\geq t] \leq 
%  \sum_{i=1}^{\Lambda}\sum_{t=1}^\gamma Pr[T\geq t\cdot i] \leq \sum_{i=1}^{\Lambda}\sum_{t=1}^\gamma Pr[T\geq \gamma i] \leq 
 \sum_{i=1}^{\Lambda/\gamma} \sum_{t=1}^\gamma
\Pr[T\geq (i+1)\gamma] 
\leq \sum_{i=2}^{\Lambda/\gamma} \frac{g(N)\gamma}{i-1} = O(g(N)\cdot\gamma\cdot\ln\Lambda)
% \\
% &\leq \sum_{t=1}^\gamma Pr[T\geq \gamma] + g(N)\gamma ln(\Lambda) \leq \gamma +g(N)\gamma ln(\Lambda) =O(\gamma g(N)ln\Lambda).
\end{align*}

Thus, replacing the value for $\gamma$, the expected runtime is at most $O(N^{\beta+r}g(N)(\phi N)^{f(N)g(N)}\ln\Lambda)$.%O(g(N)(\phi N)^{f(N)g(N)}ln\Lambda)$.
\end{proof}

\begin{cor}\label{corr:expected-quasi}
The smoothed expected time for BR to terminate for all \netcos is polynomial in $(n^{(k\log(nk))},\phi)$.
\end{cor}
\begin{proof}
From ~\eqref{quasi-delta-fixed} in Theorem \ref{thm:in-text-qpoly}, we know that the probability that the minimum improvement in a fixed BR sequence of length $5nk$ is at most $\epsilon$, is at most $\sum_{\ell=1}^{5nk}\left(2n^2k^3(\phi \epsilon)^{1/(10k\log(nk))}\right)^l$.

% Thus, the minimum improvement in any sequence of length $5nk$ is: $\sum_{l=1}^{5nk}\left(2n^2k^3(\phi \epsilon)^{1/(10klog(nk))}\right)^l$

Applying Theorem \ref{thm:expected-time}, for $N=nk$ and $\Lambda \leq k^n$, we get $f(N)=O(1)$, $N^{r+\beta}\leq N^3$, and $g(N)=O(k\log(nk))$, and the result follows.
\end{proof}

\begin{cor}\label{corr:expected-poly}
For complete graphs, the smoothed expected time for BR to terminate for \netcos is 
polynomial in $(n^k,\phi)$.
\end{cor}
\begin{proof}
From \eqref{poly-delta-fixed-n} in Lemma \ref{poly:sequence-of-length-2nk}, for the case of complete graphs when a BR sequence has all active players, we have:
\begin{equation*}
\Pr[\Delta(p)\in (0,\epsilon)]\leq \left(k^{3k}n^{2k}\phi\epsilon\right)^n\Big/\phi\epsilon \leq \sum_{i=1}^n\left(k^{3k}n^{2k}\phi^{1/2}\epsilon^{1/2}\right)^i\Big/\phi\epsilon. 
\end{equation*}

Similarly, from \eqref{poly-delta-fixed-p} in Lemma \ref{poly:sequence-of-length-2nk}, the probability that the minimum improvement in a BR sequence of length $2nk$ is at most $\epsilon$, is given by:
\begin{equation*}
\Pr[\Delta(p)\in (0,\epsilon)] \leq \sum_{p=1}^n 2\left((100\phi^2n^3k^4)^k\epsilon^{1/4}\right)^p
\end{equation*}

Combining these sums, we get the probability that a BR sequence of length $2nk$ has improvement at most $\epsilon$ is:
\begin{align*}
\Pr[\Delta(p)\in (0,\epsilon)] &\leq \max\left\{\sum_{p=1}^n 2\left((100\phi^2n^3k^5)^k\epsilon^{1/4}\right)^p, \sum_{i=1}^n\left(k^{3k}n^{2k}\phi^{1/2}\epsilon^{1/2}\right)^i\right\}\\ &\leq \sum_{j=1}^n ((\phi^c_1(nk)^{c_2k}(\phi\epsilon)^{1/c_3})^j.
\end{align*}

Applying Theorem \ref{thm:expected-time}, for $N=nk$, $N^{r+\beta}\leq N^3$, and $\Lambda \leq k^n$, we get $f(N)=O(1)$ and $g(N)=O(1)$, and the result follows.
\end{proof}

\section{Smoothness-Preserving Reduction to $1$- and $2$-\maxcut}\label{sec:red}
Recall Definition~\ref{def:SmoothRed} in Section~\ref{sec:PrelOR}, where we have defined a notion of Strong and Weak smoothness preserving reductions.
As a reminder, search problem $\CP$ is reduced to $\CQ$ if random instances of $\CP$ may be reduced to random instances of $\CQ$ in such a way that the independence of the random inputs and the bounds on their density are preserved.
The reduction is {\em strong} if independent random parameters in $\CP$ produce independent random inputs to the reduced problem, and {\em weak} if the reduced parameters are instead linear combinations of independent random variables.

Recall from Definition \ref{sec:PrelLMC} the $d$-\maxcut problem of finding a cut in a weighted graph such that the value of the cut cannot be improved by performing up to $d$ flips. In this section, we wish to provide a {\em weak} reduction from instances of $k$-\nashcoord to $2$-\maxcut. 
Solving the smoothed complexity of $2$-\maxcut is still an open problem, but it is not unlikely that it matches that of $1$-\maxcut, namely, quasi-polynomial smoothed complexity on arbitrary graphs, and polynomial smoothed complexity on complete graphs. We will show that, if this were true, then this implies the same would hold for $k$-\nashcoord, independently of the value of $k$.

We first show that a strong smoothness-preserving reduction to local-max-cut on (arbitrary) complete graphs implies (quasi-polynomial) polynomial smoothed complexity.

\begin{thm} \label{lem:reductionGivesAlg}
	Let $\CQ$ be a search problem with (quasi-)polynomial smoothed complexity, as defined above.
    Let $\CP$ be a problem which admits a strong smoothness-preserving reduction to $\CQ$, given by $f_1,\,f_2,\,f_3$, as in Definition~\ref{def:SmoothRed}. 
    Then $\CP$ has (quasi)polynomial smoothed complexity.
\end{thm}
\begin{proof}
	The algorithm for instances of $\CP$ is as follows: \\ 1) Perform the randomized reduction,\\
    2) Run the smoothed-(quasi-)polynomial-time algorithm for $\CQ$ on the reduced instance,\\
    3) Compute the solution to the instance of $\CP$ given the solution to the reduced problem.
    
    By the definition of smoothness-preserving reductions and (quasi-)polynomial smoothed complexity, step (2) will always correctly solve the reduced instance in finite time, and therefore step (3) will output a correct solution to the instance of $\CP$.
    
    It remains then to show that the algorithm runs in polynomial time with high probability, which we do via Markov's inequality.
    Let $(I,X)$ be an arbitrary instance of $\CP$ where $I$ is fixed, and $X$ is a random vector whose entries have density at most $\phi$.
    Let $R$ be a random vector whose density is at most $\phi$, and $\phi'$ be the bound on the density of the reduced random input $f_2(X,R)$. Finally, let $\CA$ be the algorithm that solves instances of $\CQ$ efficiently with high probability.
    Suppose that on random input $\big(f_1(I),f_2(X,R)\big)$,
    $\CA$ runs in time $(\phi'|I||X|)^c$ with probability $1-1/|I|^{c'}$ taken over the random input $(X,R)$, where $|X|$ is the size of $X$.
    Then we wish to show that $\CA$ runs in time $(\phi'|I||X|)^c$ on input $\big(f_1(I),f_2(X,R)\big)$, with high probability over the input vector $X$.
    To this end, we define the indicator function $B(I,X,R)$ which, for fixed values of $I$, $X$, and $R$, indicates whether $\CA$ takes time greater than $(\phi'|I||X|)^c$. 
    Thus, we wish to bound $\Pr_X[B(I,X,R)=1] = \mathbb E_{X}[B(I,X,R)]$ for all $I$ fixed, and for $R$ random. 
    Letting $\delta=1/|I|^{c'}$, we have
    \begin{align*}
    	\Pr_X\big[\,\Pr_R[B(I,X,R)=1]\geq\sqrt\delta\,\big]&=
        \Pr_X\big[\,\mathbb E[B(I,X,R)|X]\geq\sqrt\delta\,\big]\\
        (\text{Markov's Inequality})&\leq \frac{\mathbb{E}_X\big[\mathbb E_R[B(I,X,R)|X]\big]}
        {\sqrt\delta}\\
        (\text{Law of Total Expectation})&= \mathbb E_{X,R}[B(I,X,R)]/\sqrt{\delta}\\
        (\text{by assumption})&\leq \delta/\sqrt\delta=\sqrt\delta
    \end{align*}
    Therefore, with probability $1-1/|I|^{c'/2}$, the algorithm $\CA$ will solve the reduced instance in (quasi-) polynomial time. Since the values of $\phi'$, $|f_1(I)|$ and $|f_2(X,R)|$ are all assumed to be polynomial in $\phi$, $|I|$, and $|X|$, then the values $c$, $c'$ can be assumed to be constants (or logarithmic, in the quasi-polynomial case), and we have our desired result.
\end{proof}

\begin{cor}
    Let $\CP$ be a problem which admits a {\em weak} smoothness-preserving reduction to local-max-cut on an (arbitrary) complete graph, 
    then $\CP$ has (quasi-polynomial) polynomial smoothed complexity.
\end{cor}
\begin{proof}
	It suffices to show that local-max-cut has (quasi-polynomial) polynomial smoothed complexity on (arbitrary) complete graphs, when the edge weights are full-rank linear combinations of independent random variables. 
    This allows us to conclude, following the proof method of the previous lemma, that the reduction implies (quasi-)polynomial smoothed complexity for $\CP$.
To see this, we observe that the analyses of \cite{ER14,A+17} reduce to applying Lemma~\ref{lem:probability} to a high-rank collection of integer vectors, exactly as we have done in Sections \ref{sec:BRA} and \ref{sec:exp}.
Furthermore, $\ip{\alpha,MX} = \ip{M^\mathsf{T}\alpha,X}$ for any square matrix $M$. Therefore, we may restate the lemma as: \\[-0.75em]

\begin{claim}Let $X\in \R^d$ be a vector of $d$ independent random variables where each $X_i$ has density bounded by $\phi$. Let $\alpha_1,\,\dotsc,\,\alpha_r$ be $r$ linearly independent vectors in $\mathbb{Z}^d$, and $M$ a full rank matrix in $\mathbb R^{d\times d}$ with $|M_{i,j}|\geq \eta>0$ for all $i,j$ such that $M_{i,j}\neq 0$. 
then the joint density of $(\ip{\alpha_i,MX})_{i\in [k]}$ is bounded by $(\phi/\eta)^r$. In particular, for all $b_1,b_2,\dotsc\in \mathbb R$, and $\epsilon>0$,
\begin{equation*}
\Pr\Big[\textstyle\bigwedge_{i=1}^r\ip{\alpha_i,MX}\in [b_i,b_i+\epsilon]\Big] \leq (\phi\epsilon/\eta)^r
\end{equation*}
\end{claim}
The only difference is the addition of the matrix $M$. Note that as $X$ is a vector of $d$ independent random variables and $M$ is a full rank matrix, the product $MX$ is also a vector of $d$ independent random variables. Further, as every element of $X$ has density bounded by $\phi$, and every entry of $M$ is at least $\eta$, every element of the product $MX$ has density bounded by $\phi/\eta$. Hence, by applying the analysis to $MX$ instead of $X$, the proof of the claim easily follows from that of Lemma \ref{lem:probability}.

We conclude that the smoothed (quasi-polynomial) polynomial complexity for local-max-cut on (arbitrary) complete graphs from \cite{ER14,A+17} does not require edge-weights to be independent, but instead, it suffices to have edge weights which are full-rank, linear combinations of independent random variables, where the non-zero entries of $M$ are bounded away from 0.
	Since the smoothed analysis of local-max-cut needs $\epsilon$ to be $1/poly(|X|,\phi)$, it suffices to have $\eta \geq 1/poly(\phi,|X|)$.
    
	The rest of the analysis is identical to that of Theorem~\ref{lem:reductionGivesAlg}, where $\CA$ is the FLIP algorithm.
\end{proof}

We note that if it were possible to weakly reduce $k$-\nashcoord to $1$-\maxcut, then this would imply a (quasi-)polynomial smoothed complexity for $k$-\nashcoord, where the degree of the polynomial does not depend on $k$. 
Unfortunately, we only achieve a weak reduction to $2$-\maxcut, which is likely to have similar smoothed complexity to $1$-\maxcut, though this is not as of yet known. 
We leave the smoothed analysis of $2$-\maxcut, and therefore of \nashcoord for $k$ variable, as an open problem.

\begin{thm}\label{thm:reduction}
	The problem of finding a Nash Equilibrium in a Network Coordination Game with $k$ strategies {\em ($k$-\nashcoord)} admits a {\em weak} smoothness-preserving reduction to local-max-cut up to two-flips {\em (2-\maxcut)}. Furthermore, {\em 2-\nashcoord} reduces to {\em 1-\maxcut}.
\end{thm}
\begin{proof}
	Assume, first, that the payoff values of the coordination game are supported in $[0.5,1]$. If the input is assumed to have been supported on $[-1,1]$, then this is simply an affine transformation of the input, and at most quadruples the maximum density.
    The idea of the reduction is to set up a graph such that every ``good'' cut can be mapped to a strategy profile, the total cut values of these ``good'' cuts are equivalent to the payoff of the associated strategy profile, and every locally maximal cut must be a ``good'' cut.
    Recall that the definition of smoothness-preserving reduction allows for extra randomness to be introduced.
    We will use this randomness to ensure that these conditions hold, and that the edge weights are simply a full-rank, integer combination of the payoff values and the extra random variables.
     Since the cut graph topology is only a function of the game graph topology, and the edge weights are only a function of the input payoff values and the extra random variables, this will be a valid reduction. It will suffice to argue that the density bound does not blow up, and that local max cuts optimal up to two flips are exactly those ``good'' cuts which are mapped to PNE strategy profiles.
   
   In this proof, the total payoff function for a strategy profile $\sss$ will be considered as 
   \[\payoff(\sss)=\sum_{u}\payoff_{u}(\sss) = 2\sum_{uv\in E}A_{uv}(\sigma_u,\sigma_v)
   \]
   which is double the potential function considered in the previous sections. 
   This will be necessary to ensure that the linear system has integer entries. 
   To have our payoff value be equal to that of the standard potential function, it suffices to halve the payoff values, which at most doubles the density of the random variables. (Alternatively, we may set $\eta=\tfrac12$ in the proof of Theorem~\ref{lem:reductionGivesAlg}.)
   
   The reduction is as follows: given an instance of $k$-\nashcoord on game graph $G$, we construct a graph with $nk+2$ nodes: two terminal nodes $s$ and $t$, and $nk$ nodes indexed by player-strategy pairs $(u,i)$. 
Nodes $s$ and $t$ are connected to every other node in the graph, and for each player $u$ and strategy $i$, there is an edge from node $(u,i)$ to node $(u,j)$ for all $j\neq i$.
Furthermore, if players $u$ and $v$ share an edge in the game graph ($i.e.$ play a game together), there is also an edge from $(u,i)$ to $(v,j)$ for all $1\leq i,j\leq k$.
Therefore, the cut graph is complete if and only if the game graph is.

Call a cut $S,T$ {\em valid} if it is an $s$-$t$ cut with $s\in S$ and $t\in T$, and $S$ contains at most one node $(u,i)$ for each player $u$.
Now, for any valid $s$-$t$ cut, we can interpret this cut as determining a strategy profile as follows:
If player $u$ appears in $S$ paired with strategy $i$, then set $\sigma_u=i$. 
Otherwise, set $\sigma_u=0$, a ``dummy'' strategy with bad payoff. Call this profile $\sss(S)$.
We wish to choose edge weights such that all locally maximal cuts are valid cuts, 
and also such that for any valid cut $S$, the total cut value is equal to $\payoff(\sss(S))$.
We denote as $A((u,i)(v,j))$ the payoff value for the $uv$ game, when player $u$ plays strategy $0\leq i\leq k$, and player $v$ plays strategy $0\leq j\leq k$. This is simply to disambiguate the $A_{uv}(i,j)$ notation. We will, however, use this latter notation when space does not permit the former. 
Letting 0 denote the dummy strategy, we assume that $A((u,i)(v,0))=\underline Y(u,i)$ for all $v\neq u$, and $A((u,0)(v,0))=\underline A_0$ for all $u\neq v$, where
the underlined values denote new random variables which are not given by the instance of $k$-\nashcoord.

   Since the variables $\underline A_0$ and $\underline Y(u,i)$ are in our control, we assume that they are drawn independently at random, and are supported on $[0,0.5)$. 
   This ensures that any Nash Equilibrium must entirely consist of non-zero strategies, since we have assumed $A((u,i)(v,j))\in[0.5,1]$ for all $1\leq i,j\leq k$.
   To minimize the density of the new random variables, we may assume that the distribution is uniform.
   For the purposes of the cut graph, we will also need random variables $\underline W(u,i)$ for all players $u$ and $1\leq i\leq k$, and $\underline R(u,ij)$ for all $u$ and $1\leq i<j\leq k$.
   We assume the $\underline W$ and $\underline R$ variables to be $i.i.d.$ uniform random variables over $[-1,0)$.

	Denote as $\delta:2^{V'}\to\mathbb R$ is the cut function, where $\delta(S):=\sum\limits_{\substack{uv\in E':\\u\in S, v\notin S}} w(uv)$.
   Furthermore, for simplicity of notation, let $\pi(\{(u_1,i_1),\,\dotsc,\,(u_\ell,i_\ell)\})$ denote the value: $\payoff\big(\sss\big(\{s,(u_1,i_1),\,\dotsc,\,(u_\ell,i_\ell)\}\big)\big)$, where $\pi$ stands for ``payoff.'' 
   We wish to choose edge weights $w$ such that
   \begin{enumerate}
   \item[(i)] For every valid cut $S,T$, $\delta(S)=\payoff(\sss(S))=\pi(S\setminus\{s\})$,
   \item[(ii)] for every player $u$ and $1\leq i<j\leq k$, 
   $\delta(\{s,(u,i),(u,j)\}) = 2\underline R(u,ij)$,
   \item[(iii)] and for every pair $(u,i)\in \mathcal S$,
   $w((s,(u,i)))=\underline W(u,i)$.
   \end{enumerate}
   
   The rest of the proof is contained in the 3 following claims, which we will prove at the end of this section:
   \begin{claim}\label{claim1}
   Condition (i) is satisfied if and only if (a) $\delta(\{s,(u,i)\})=\pi(\{(u,i)\})$ for all players $u$ and $1\leq i\leq k$, and (b) $w((u,i),(v,j))= \underline Y(u,i)-\underline Y(v,j)-A((u,i)(v,j))-\underline A_0$.
   \end{claim}
   \begin{claim}\label{claim2}
   The edge weights $w(u,v)$ which satisfy conditions (i), (ii), and (iii) are full-rank, square, integer-valued, linear combinations of the random variables $A_{uv}(i,j)$, $\underline Y(u,i)$, $\underline A_0$, $\underline R(u,ij)$, and $\underline W(u,i)$. %for all choices of $u,\,v,\,i,\,j$.
   \end{claim}
   \begin{claim}\label{claim3}
    	If conditions (i), (ii), and (iii) are satisfied, then all local-max-cuts up to $2$ flips are valid cuts, and their associated strategy profiles are Nash equilibria. 
    \end{claim}
    
    Thus, we have provided a reduction from an instance of $k$-\nashcoord to an instance of 2-\maxcut, such that any solution to the reduced instance of 2-\maxcut directly translates to a solution to the original instance of $k$-\nashcoord, the edge weights of the cut graph are a full-rank linear combination of the game payoff values and the extra random variables,
    and the linear combinations are integral.
    Therefore, we have provided a weak smoothness-preserving reduction from $k$-\nashcoord on complete (resp. arbitrary) game graphs to $2$-\maxcut on complete (arbitrary) graphs.
    
   In the case $k=2$, we slightly modify the reduction to not include a dummy strategy. Simply create a graph on $n+2$ nodes labeled as $s$, $t$, and one for each player in the Network Game. 
   We say that every $s$-$t$ cut $S,T$ is valid, 
   and define $\sss(S)$ as setting $\sigma_u=1$ if $u\in S$, and $\sigma_u=2$ if $u\in T$.
   Then $w(s,u)=\underline W(u)$, $w(u,v)=A((u,1)(v,2))+A((u,2)(v,1))-A((u,1)(v,1))-A((u,2)(v,2))$, and the rest of the proof goes through.
   Since the default (``no node selected'') strategy is one of the two possible strategies, every local max cut up to $1$ flip must map to a PNE, as desired.\end{proof}   
   \begin{cor}
   	If {\em 2-\maxcut} has (quasi-polynomial) polynomial smoothed complexity on (arbitrary) complete graphs when inputs are linearly independent combinations of independent random variables, then {\em \nashcoord} has (quasi-polynomial) polynomial smoothed complexity on (arbitrary) complete game graphs for $k$ in the input, rather than for fixed $k$.
   \end{cor}
   
   We leave as an open problem the smoothed complexity of $2$-\maxcut.
   Below are presented the proofs of the claims necessary for Theorem~\ref{thm:reduction}

   \begin{proof}[Proof of Claim~\ref{claim1}:]
   \textit{Condition (i) is satisfied if and only if (a) $\delta(\{s,(u,i)\})=\pi(\{(u,i)\})$ for all players $u$ and $1\leq i\leq k$, and (b) $w((u,i),(v,j))= A((u,i)(v,j))+\underline A_0-\underline Y(u,i)-\underline Y(v,j)$.}
   
   	We begin by showing the two following equalities:
   \begin{align}
		\pi(\underbrace{\{(u_1,i_1),\dotsc,(u_\ell,i_\ell)\}}_{S'}) &= 
		\Big[\sum_{j=1}^\ell \pi(\{(u_j,i_j)\})\Big] - (\ell-1) \pi(\emptyset)\nonumber\\
		&\quad - \sum_{\substack{(u,i),(v,j)\in S\\uv\in E}} 
        \mkern-18mu
			2\left[A_{uv}(i,0)+A_{uv}(0,j)-A_{uv}(0,0)-A_{uv}(i,j)\right]\label{eq:H-breakdown}\\
		\delta(\underbrace{\{s,(u_1,i_1),\dotsc,(u_\ell,i_\ell)\}}_S) &= 
		\Big[\sum_{j=1}^\ell\delta(\{s,(u_j,i_j)\})\Big] - (\ell-1)\delta(\{s\})\nonumber\\
		&\quad - \sum_{\substack{(u,i),(v,j)\in S\\uv\in E}} 
      	\mkern-18mu 2w((u,i),(v,j)) \label{eq:w-breakdown}
	\end{align}
    For (\ref{eq:H-breakdown}), note first if there is no $uv$ edge in the game graph, then $A_{uv}$ does not appear on either side of the equality, and we may restrict our attention to pairs which form game edges. 
    Now, for every $v$ and $w$ which do not appear in $S$, the left-hand-side has $2A_{vw}(0,0)$, and the right-hand-side has $2(\ell-(\ell-1))A_{vw}(0,0)$ from the first line. If $u$ appears with strategy $i$, and $v$ does not appear in $S$, then the left-hand-side has $2A_{uv}(i,0)$, and the right-hand-side has $2A_{uv}(i,0)$ from the $\pi(\{(u,i)\})$ term. 
    If $u$ appears with strategy $i$, and $v$ appears with strategy $j$, then the left-hand-side has $2A_{uv}(i,j)$, and the right-hand-side has $2A_{uv}(i,0)$ and $2A_{uv}(0,j)$ from the $\pi(\{(u,i)\})$ and $\pi(\{(v,j)\})$ terms which are canceled out by the second line, $2(\ell-2-(\ell-1))A_{uv}(0,0)$ terms from the first line which is canceled out by the second line, and the term $2A_{uv}(i,j)$ from the second line.
    
    A similar argument shows the validity of (\ref{eq:w-breakdown}). Since condition (i) requires that $\pi(\{u,i\})=\delta(\{s,(u,i)\})$, this is necessary. 
    In the case $\ell=2$, this implies that $w((u,i)(v,j))$ must be equal to $A_{uv}(i,0)+A_{uv}(0,j)-A_{uv}(0,0)-A_{uv}(i,j)\ = \ \underline Y(u,i)+\underline Y(u,j)-\underline A_0 -A_{uv}(i,j)$. Finally, setting $w((u,i)(v,j))$ to this value fulfills condition (i) for all values of $\ell$, as desired.
  	\end{proof}
%     Thus, if $H(\{(u,i)\})=\delta(\{s,(u,i)\})$ for all players $u$ and $1\leq i\leq k$, then the $\ell=2$ requires us to set  $w((u,i)(v,j))=-\tfrac12\big(Y(u,i)+Y(v,j)-X_0-X((u,i)(v,j))$.

   \begin{proof}[Proof of Claim~\ref{claim2}:] \textit{The edge weights which satisfy conditions (i), (ii), and (iii) are a full-rank, square, integer-valued, linear combination of the random variables $A_{uv}(i,j)$, $\underline Y(u,i)$, $\underline A_0$, $\underline R(u,i,j)$, and $\underline W(u,i)$ for all choices of $u,\,v,\,i,\,j$.}
   
    The previous claim allows us to set up the following system:
    \begin{align}
    	w(s,(u,i))&=\underline W(u,i)\\[2ex]
        w((u,i)(v,j))&= \underline Y(u,i)+\underline Y(v,j)-A_{uv}(i,j)-\underline A_0\\
        \pi(\emptyset)=\delta(\{s\})&= w(s,t) + \sum_{u\in V}\sum_{i=1}^k \underline W(u,i)\nonumber\\
        \implies w(s,t)&= \pi(\emptyset) - \sum_{u\in V}\sum_{i=1}^k \underline W(u,i)\label{eq:st-edge}\\[2ex]
        2\underline R(u,ij)=\delta(\{s,(u,i),(u,j)\})&=
        \delta(\{s,(u,i)\})+\delta(\{s,(u,j)\})-\delta(\{s\})
        -2w((u,i)(u,j))\nonumber\\
        \implies w((u,i)(u,j))&= \tfrac12\big(\,
        	\pi(\{(u,i)\})+\pi(\{(u,j)\})-\pi(\emptyset)-2\underline R(u,i,j)
        \,\big)
        \end{align}
        \begin{align}
        \pi(\{(u,i)\})=\delta(\{s,(u,i)\})&= 
        w(s,t)+w((u,i),t)+\sum_{(v,j)\neq(u,i)} [\underline W(v,j) + w((u,i)(v,j))]\nonumber\\
        \implies w((u,i),t)&=\pi(\{(u,i\})-w(s,t) - \sum_{(u,i)\neq (v,j)} [\underline W(v,j)
        + w((u,i)(v,j))]\nonumber\\
        &=\pi(\{(u,i\})-\pi(\emptyset) + \underline W(u,i) - \sum_{(u,i)\neq (v,j)} w((u,i)(v,j))\label{eq:ut-edges}
    \end{align}
    We observe first that (\ref{eq:ut-edges}) adds the values of the previous numbered equations to the value of $w((u,i),t)$. Therefore, it suffices to perform simple row-elimination to get $\widehat w((u,i),t)=\pi(\{u,i\})-\sum \underline W(v,j)$. Now, let $G=(V,E)$ be the underlying game graph, and let $d(u)$ be the number of games that player $u$ participates in, $i.e.$ the degree of $u$ in the game graph. Then $\pi(\emptyset)=2|E|A_0$, and $\pi(\{(u,i)\})=\pi(\emptyset) + 2d(u)[\underline{Y}(u,i)-\underline A_0]$.
    Finally, we have 
    
    \begin{equation}\def\arraystretch{1}
	\left(\begin{matrix}
    	\vdots\\w((u,i)(v,j))\\\vdots\\\hline
    	\vdots\\w((u,i)(u,j))\\\vdots\\\hline
        \vdots\\\widehat w((u,i),t)\\\vdots\\\hline
        w(s,t)\\\hline
        \vdots\\ w(s,(u,i))\\\vdots
    \end{matrix}\right)=\def\arraystretch{2}
    \left(\begin{array}{c|c|c|c|c}
    	-Id&\bm 0&*&-\bm 1&\bm 0\\\hline
        \bm 0&-Id&*&*&\bm 0\\\hline
        \bm 0&\bm 0&2d(u)Id&*&Id\\\hline 
        \bm 0&\bm 0&\bm 0&2|E|&-\bm 1\\\hline 
        \bm 0&\bm 0&\bm 0&\bm 0&Id
    \end{array}
    \right)\def\arraystretch{1}
    \left(\begin{matrix}
    	\vdots\\A_{uv}(i,j)\\\vdots\\\hline
    	\vdots\\\underline R(u,i,j)\\\vdots\\\hline
        \vdots\\\underline Y(u,i)\\\vdots\\\hline
        \underline A_0\\\hline
        \vdots\\ \underline W(u,i)\\\vdots
    \end{matrix}\right)
\end{equation}
It is easy to check that the $*$ values are integral, since the $\pi$ values must be even combinations of the $A$ values. Therefore, after the row-operations leading to $\widehat w((u,i),t)$ values, the matrix is upper-triangular, which implies that the system is full-rank, square, and integral, as desired.
    \end{proof}
    
    \begin{proof}[Proof of Claim~\ref{claim3}:]
    \textit{If conditions (i), (ii), and (iii) are satisfied, then all local-max-cuts up to $2$ flips are valid cuts, and their associated strategy profiles are Nash equilibria.}
    
    Recall that we have assumed that $0.5\leq A_{uv}(i,j)\leq 1$ for all edges $uv$ and for all $1\leq i,j\leq k$, that $0\leq \underline A_0,\underline Y(u,i)<0.5$ for all players $u$ and $1\leq i\leq k$, and that $-1\leq \underline R(u,ij),\underline W(u,i)<0$ for all players $u$ and $1\leq i<j\leq k$.
    
    	We wish to show, first, that any local max cut must be an $s$-$t$ cut. Without loss of generality, assume that all cuts considered do not include $t$, since it suffices to take the complement of the cut set.
        Thus, it suffices to argue that for any set $S$ of vertices containing neither $s$ nor $t$, 
        that $\delta(S\cup\{s\})-\delta(S)>0$.
        The two cuts may only differ on edges incident to $s$. The positive term includes $w(s,t)$ and $\underline W(u,i)$ for all $(u,i)\notin S$, and the negative term includes $\underline W(u,i)$ for all $(u,i)\in S$.
        However, we know from (\ref{eq:st-edge}) that $w(s,t)=\pi(\emptyset)-\sum_{(u,i)}\underline W(u,i)$. 
        Therefore, we get
        \[
        	\delta(S\cup\{s\})-\delta(S) \ \ =\ \  
            2|E|\underbrace{\underline A_0}_{\geq 0}
            \ -\  2\!\!\!\sum_{(u,i)\in S} 
            \underbrace{\underline W(u,i)}_{< 0}\ \ >\ \ 0
        \]
        
%     	Note that $w(s,t)=H(\emptyset)-\sum_{(u,i)} \underline W(u,i)\geq -\sum W(u,i)$. 
%         Let $S$ be any set of player-strategy pairs. 
%         Then $\delta(S)-\delta(\{s\}\cup S) \leq  \sum_{u\in V}\sum_{i=1}^k \underline W(u,i)$, since every $w(s,(u,i))$ edge is either subtracted in the $w(s,t)$ term of $\delta(\{s\}\cup S)$ if $(u,i)\notin S$, or added in the $\delta(S)$ term otherwise. 
%         But $W(u,i)<0$ for all $u,i$. Thus, any local max cut must contain $s$. 
%         If $S$ contains $t$, the $\underline W$ values cancel, and the difference becomes $-H(\emptyset)<0$.
%         Therefore, any local max cut must contain $s$.
        
%         \note{Show that any cut which contains $s$ shouldn't contain $t$. We get a similar subtraction of the W's... }
%         similar argument \note{(details?)} shows that any local max cut must not contain $t$, and therefore it must be an $s$-$t$ cut.
        
    	Now, it remains to show that any locally optimal $s$-$t$ cut must be valid. It is well known that cut functions in undirected graphs are submodular, and therefore for any $S$ not containing $s$, $(u,i)$, or $(u,j)$, \begin{align*}
        	\delta(S\cup &\{s,(u,i),(u,j)\}) - 
            \delta(S\cup \{s,(u,i)\})\\&\leq 
            \delta(\{s,(u,i),(u,j)\})-
            \delta(\{s,(u,i)\})\\
            &= 2\underline R(u,i,j)-2(|E|-d(u))\underline A_0-2d(u) \underline Y(u,i)\\
            &<0
        \end{align*}
        Thus, conditions (ii) and (iii) are sufficient to guarantee that all local max cuts are valid cuts, and therefore local max cuts up to $2$ flips must also be valid cuts.
        Condition (i) implies that the {\em value} of a valid cut is equal to its associated strategy profile. 
        However, for any strategy profile $\sigma$, if $S$ is the valid cut associated to $\sigma$, and player $u$ benefits from replacing $\sigma(u)$ with $i'$, then this implies that $S-(u,\sigma(u))+(u,i')$ is a cut with greater value, which is a $2$-flip move. Therefore, $S$ can only be a max-cut up to $2$ flips if its associated strategy profile forms a Nash Equilibrium, and by construction, this strategy profile must consist of non-zero strategies.      
    \end{proof}

\end{document}